\documentclass[pdflatex,sn-mathphys-ay]{sn-jnl}

\usepackage{graphicx}%
\usepackage{multirow}%
\usepackage{amssymb}
\usepackage{amsmath}
\usepackage{amsfonts}
\usepackage{mathrsfs}
\usepackage{mathtools}
\usepackage[title]{appendix}%
\usepackage{xcolor}%
\usepackage{textcomp}%
\usepackage{centernot}
\usepackage{manyfoot}%
\usepackage{booktabs}%
\usepackage{algorithm}%
\usepackage{algorithmicx}%
\usepackage{algpseudocode}%
\usepackage{listings}%
\usepackage{enumitem}

\usepackage{tikz}
\usetikzlibrary{backgrounds,arrows.meta,positioning,fit,calc,decorations.pathreplacing,shapes.multipart,shapes.geometric}
\usepackage{cleveref}

\theoremstyle{thmstyleone}
\newtheorem{theorem}{Theorem}[section]
\newtheorem{proposition}[theorem]{Proposition}
\newtheorem{lemma}[theorem]{Lemma}

\theoremstyle{definition}

\newtheorem{assumption}[theorem]{Assumption}
\theoremstyle{remark}
\newtheorem{remark}[theorem]{Remark}

\begin{document}

\title{Beyond Diagonal Noise: A Better Predator-Prey Modeling Framework with Cross-Covariance}

\author[1]{\fnm{Jiguang} \sur{Yu}}\email{jyu678@bu.edu}
\equalcont{These authors contributed equally to this work as co-first authors.}

\author*[2]{\fnm{Louis Shuo} \sur{Wang}}\email{swang116@vols.utk.edu}
\equalcont{These authors contributed equally to this work as co-first authors.}

\affil[1]{\orgdiv{College of Engineering},
  \orgname{Boston University},
  \orgaddress{\city{Boston}, \postcode{02215}, \state{MA}, \country{United States}}}

\affil[2]{\orgdiv{Department of Mathematics},
  \orgname{University of Tennessee},
  \orgaddress{\city{Knoxville}, \postcode{37996}, \state{TN}, \country{United States}}}

\abstract{
The introduction of stochasticity into continuous ecological models frequently relies on phenomenological, diagonal diffusion terms that lack a rigorous microscopic basis. We demonstrate that this standard practice fundamentally misrepresents the geometry of demographic fluctuations. By deriving a stochastic Rosenzweig--MacArthur model directly from an integer-valued, Bernoulli-coupled continuous-time Markov chain, we isolate the exact diffusion covariance structure dictated by event stoichiometry. We mathematically prove that coupled predation--conversion events inherently generate a structurally negative predator--prey cross-covariance, exposing the severe mathematical and biological limitations of standard diagonal-noise approximations. Furthermore, we resolve a persistent ambiguity in stochastic population modeling by explicitly formalizing the bifurcation between open-domain formulations (for survival-conditioned interior dynamics) and absorbed formulations (for extinction-permitting dynamics). To rigorously support this distinction, we develop a tailored two-stage Lyapunov well-posedness architecture that separates non-explosion criteria from boundary-barrier positivity invariance. By bridging microscopic event stoichiometry with macroscopic boundary-degenerate diffusions, this work replaces ad hoc noise constructs with a definitive, mathematically exact template for covariance-consistent and boundary-aware ecological modeling.}


\maketitle

\section{Introduction}
\label{sec:introduction_paperA}

Predator--prey interactions constitute one of the foundational subjects of mathematical ecology. In particular, the Rosenzweig--MacArthur (R--M) model with Holling type~II functional response remains a central benchmark for studying coexistence, oscillatory dynamics, and the paradox of enrichment \citep{rosenzweig_graphical_1963,rosenzweig_paradox_1971,may_stability_2001,beay_stability_2019}. 
In the deterministic setting, the R--M model exhibits a rich bifurcation structure. As the prey carrying capacity increases, the stable coexistence equilibrium loses stability through a Hopf bifurcation, which produces sustained limit cycles and increases the risk of low-density excursions. 
This mechanism, in which enrichment destabilizes coexistence, has attracted sustained attention in both ecological theory and applied population biology \citep{gilpin_enriched_1972,arditi_coupling_1989,roy_stability_2007}.

\begin{figure*}[p]
\centering
\begin{tikzpicture}[
    font=\small,
    scale=0.55,
    transform shape,
    >=Latex,
    node distance=8mm and 10mm,
    line/.style={-Latex, thick},
    dashedline/.style={-Latex, thick, dashed},
    box/.style={
        draw, rounded corners=2.5mm, thick,
        align=left, inner sep=4mm,
        minimum height=10mm,
        text width=4.4cm
    },
    widebox/.style={
        draw, rounded corners=2.5mm, thick,
        align=left, inner sep=4mm,
        minimum height=10mm,
        text width=9.4cm
    },
    pill/.style={
        draw, rounded corners=6mm, thick, align=center,
        inner sep=2.5mm, minimum height=8mm
    },
    decision/.style={
        diamond, draw, thick, aspect=2.0,
        align=center, inner sep=1.8mm, text width=3.2cm
    },
    smallnote/.style={align=left, inner sep=1.5mm, font=\scriptsize},
    titlebox/.style={
        draw, rounded corners=3mm, thick, fill=black!3,
        align=center, inner sep=3mm, text width=14.4cm
    }
]

\node[titlebox] (title) {
\textbf{Mechanistic stochastic R--M framework:}
from event-level CTMC/CME to full-covariance diffusion, boundary-aware formulations, and two-stage well-posedness
};

\node[box, fill=blue!6, below=10mm of title, xshift=-5.0cm] (det) {
\textbf{Deterministic backbone (R--M + Holling II)}\\
\(\dot N = N(1-N/k)-\frac{mNP}{1+N}\),\quad
\(\dot P = P(-c+\frac{mN}{1+N})\)\\
Coexistence \(K_3\), Hopf threshold \(k_H=\frac{m+c}{m-c}\)\\
\(\Lambda_2\): stable coexistence,\quad \(\Lambda_1\): oscillatory regime
};

\node[box, fill=green!6, right=10mm of det, text width=4.8cm] (motivation) {
\textbf{Motivation near Hopf}\\
Deterministic restoring force weakens\\
\(\Rightarrow\) demographic fluctuations amplified\\
\(\Rightarrow\) quasi-cycles / extinction risk elevated
};

\node[widebox, fill=orange!8, below=10mm of det, xshift=2.55cm] (ctmc) {
\textbf{Mechanistic count-level model (CTMC \(\rightarrow\) CME)}\\
State \(X=(X_N,X_P)\in\mathbb N_0^2\), density \(x=X/\Omega\).\\
Bernoulli-coupled predation--conversion channels (integer-valued jumps): prey birth, competition death, predator death, predation with/without conversion.\\
Exact generator/CME implies count-level infinitesimal covariance:
\[
G(X)=\sum_k \lambda_k(X)\,\nu_k\nu_k^\top
= S\,\mathrm{diag}(\lambda(X))\,S^\top.
\]
};

\node[widebox, fill=purple!7, below=9mm of ctmc] (cle) {
\textbf{CME-consistent diffusion (CLE) and covariance identity}\\
Under density-dependent scaling \(\lambda_k(X)=\Omega f_k(X/\Omega)\):
\[
dx_t = b(x_t)\,dt+\Sigma(x_t)\,dW_t,\qquad
b(x)=Sf(x),
\]
\[
a(x)=\Sigma(x)\Sigma(x)^\top
=\frac{1}{\Omega}\,S\,\mathrm{diag}(f(x))\,S^\top.
\]
\textbf{Key principle:} drift equivalence \(\not\Rightarrow\) covariance equivalence (channel decomposition matters at second order).
};

\node[box, fill=red!8, below left=9mm and 0mm of cle, text width=6cm] (coupled) {
\textbf{Coupled predation closure (mechanistic baseline)}\\
Same-event prey loss + predator gain\\
\[
(a_{\mathrm{pred}})_{12}(x)
=
-\frac{1}{\Omega}e\frac{mNP}{1+N}<0
\quad (x\in D)
\]
\textbf{Structural negative cross-covariance}\\
\(\Rightarrow\) tilted local covariance ellipses
};

\node[box, fill=gray!10, below=9mm of cle, text width=4.7cm] (split) {
\textbf{Split-channel comparator (diagonal-noise proxy)}\\
Prey removal and predator reproduction split into independent channels\\
\[
(a_{\mathrm{pred}}^{(S)})_{12}(x)=0
\]
Matches drift, misses fluctuation geometry\\
(axis-aligned ellipses)
};

\node[box, fill=yellow!10, below right=9mm and 0mm of cle, text width=4.7cm] (diag) {
\textbf{Practical diagnostic for diagonalization}\\
Only use diagonal-noise approximation when
\[
\rho(x)=\frac{|a_{12}(x)|}{\sqrt{a_{11}(x)a_{22}(x)}}\ll 1
\]
(e.g. low \(e\), weak predation, restricted region)
};

\node[decision, fill=cyan!8, below=12mm of split, yshift=-1mm] (branch) {
Same local\\coefficients \(b,a\),\\but which global\\formulation?
};

\node[box, fill=cyan!6, below left=10mm and 13mm of branch, text width=5.0cm] (open) {
\textbf{Open-domain formulation on \(D=(0,\infty)^2\)}\\
Interior coexistence dynamics (survival-conditioned)\\
Targets: local fluctuation geometry, linear noise approximation (LNA)/stochastic sensitivity function (SSF), spectral structure, near-\(K_3\) analysis
};

\node[box, fill=teal!8, below right=10mm and 13mm of branch, text width=5.0cm] (absorb) {
\textbf{Absorbed formulation on \(Q\)}\\
Freeze at first boundary hit \(\tau_{\partial D}\)\\
Targets: extinction probabilities, persistence times, quasi-stationary behavior\\
(CTMC extinction irreversibility analogue)
};

\node[widebox, fill=green!8, below=10mm of open, xshift=2.65cm] (wp) {
\textbf{Two-stage well-posedness architecture (for open-domain SDE)}\\
\textbf{Stage 1:} maximal strong solution + no interior blow-up before boundary hitting (Lyapunov anti-explosion)\\
\textbf{Stage 2:} positivity invariance (\(\tau_{\partial D}=\infty\) a.s.) only under an additional boundary-barrier Lyapunov condition\\
Message: positivity is not baked into baseline well-posedness; it is a separate property aligned with modeling purpose.
};

\node[widebox, fill=blue!4, below=10mm of wp] (num) {
\textbf{Numerical diagnostics / visual consequences}\\
stochastic simulation algorithm (SSA) vs ODE fluctuation scale \(\sim O(\Omega^{-1/2})\) \quad $\bullet$ \quad
Covariance ellipses (tilted vs axis-aligned) \quad $\bullet$ \quad
Cross-covariance heatmap \(\rho(K_3)\) \quad $\bullet$ \quad
Absorbed vs open-domain boundary behavior \quad $\bullet$ \quad
Regime organization relative to Hopf threshold (\(\Lambda_2,\Lambda_1\))
};

\node[pill, fill=red!6, below=10mm of num, text width=14.2cm] (takehome) {
\textbf{Take-home message:}
In stochastic predator--prey diffusion models, covariance is a mechanistic output of event stoichiometry (not a stylistic noise choice); 
boundary treatment defines a modeling bifurcation; and well-posedness should be stated in a boundary-aware, two-stage architecture.
};

\draw[line] (det.east) -- (motivation.west);
\draw[line] (det.south) |- ([xshift=-4.9cm]ctmc.north);
\draw[line] (motivation.south) |- ([xshift=4.9cm]ctmc.north);

\draw[line] (ctmc.south) -- (cle.north);

\draw[line] (cle.south west) -- (coupled.north);
\draw[line] (cle.south) -- (split.north);
\draw[line] (cle.south east) -- (diag.north);

\draw[line] (coupled.south) |- (branch.west);
\draw[line] (split.south) -- (branch.north);
\draw[line] (diag.south) |- (branch.east);

\draw[line] (branch.south west) -- (open.north);
\draw[line] (branch.south east) -- (absorb.north);

\draw[line] (open.south) |- ([xshift=-4.6cm]wp.north);
\draw[dashedline, color=teal!70!black] (absorb.south) |- ([xshift=4.6cm]wp.north)
    node[pos=0.28, right, smallnote] {different observables \\ same local coefficients};

\draw[line] (wp.south) -- (num.north);
\draw[line] (num.south) -- (takehome.north);

\node[smallnote, anchor=south west, text=red!70!black] at ($(coupled.north west)+(1mm,1mm)$) {\textbf{C1 + C2}};
\node[smallnote, anchor=south west, text=cyan!60!black] at ($(branch.north west)+(1mm,1mm)$) {\textbf{C3}};
\node[smallnote, anchor=south west, text=green!50!black] at ($(wp.north west)+(1mm,1mm)$) {\textbf{C4}};

\begin{scope}[on background layer]
\node[draw=orange!45, rounded corners=3mm, thick, inner sep=3mm,
      fit=(ctmc)(cle)(coupled)(split)(diag),
      fill=orange!2] {};
\node[draw=cyan!40!black, rounded corners=3mm, thick, inner sep=3mm,
      fit=(branch)(open)(absorb)(wp),
      fill=cyan!2] {};
\end{scope}

\end{tikzpicture}

\caption{\textbf{Mechanistic roadmap and modeling bifurcation for the stochastic R--M framework.}
The diagram summarizes the paper's logic from the deterministic R--M backbone and Hopf organization, through a mechanistic CTMC/CME derivation of a full-covariance CLE, to the structural negative predator--prey cross-covariance induced by coupled predation--conversion. It then highlights the explicit modeling bifurcation between open-domain and absorbed diffusion formulations built from the same local coefficients, and the corresponding two-stage well-posedness architecture for the open-domain SDE.}
\label{fig:mechanism_summary_tikz}
\end{figure*}
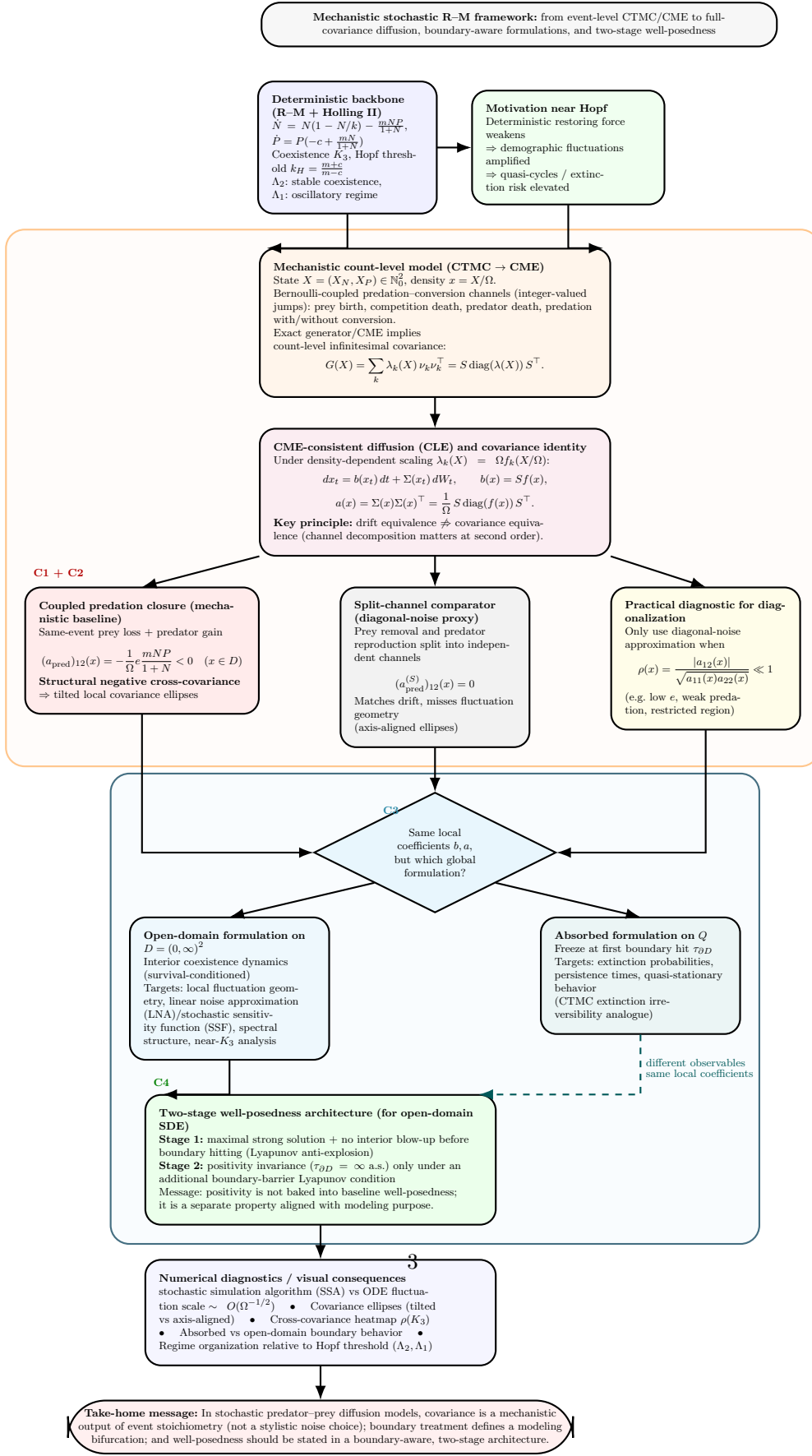

For finite populations, however, the deterministic ordinary differential equation (ODE) is merely a mean-field approximation. 
Demographic stochasticity, which arises from randomness in individual birth, death, and predation events, introduces fluctuations that can fundamentally alter quantitative predictions. 
Even when the deterministic coexistence equilibrium is locally stable, demographic noise can generate quasi-cycles and facilitate noise-induced transitions toward extinction \citep{mckane_predator-prey_2005,pinedakrch_tale_2007,black_stochastic_2012}. 
These phenomena are especially pronounced near the Hopf threshold, where the deterministic restoring force weakens and fluctuation amplification becomes large.
A rigorous analysis of these dynamics requires a stochastic description that is both mechanistically faithful to the underlying individual-level events and mathematically tractable for asymptotic study.

In this paper, we develop precisely such a description. 
We begin with an integer-valued continuous-time Markov chain (CTMC) formulation of the R--M system that includes explicit Bernoulli-coupled predation--conversion events. From this formulation, we derive a chemical master equation (CME) and, subsequently, a demographic diffusion approximation in the form of a chemical Langevin equation (CLE).
Crucially, the drift and covariance of this CLE are strictly dictated by event stoichiometry and channel intensities. 
A central mathematical consequence of this derivation is that the mechanistically correct diffusion covariance is inherently non-diagonal: coupled predation--conversion events induce a structurally negative predator--prey cross-covariance.
This geometric feature is absent from the phenomenological diagonal or multiplicative noise constructions that dominate the literature.
We prove that this cross-covariance is a direct algebraic consequence of the stoichiometric coupling between prey loss and predator gain within the same demographic event, rather than an ad hoc modeling embellishment.

Beyond the mechanistic derivation of the noise structure, a primary aim of this work is to address a persistent source of ambiguity in stochastic predator--prey modeling. This ambiguity arises from the routine conflation of two distinct diffusion formulations built from the same local stochastic differential equation (SDE) coefficients. 
We establish that an open-domain formulation on the interior $(0,\infty)^2$ provides the appropriate framework for studying interior coexistence dynamics. This framework captures fluctuation geometry near equilibrium, survival-conditioned statistics, and local spectral structure. 
Conversely, an absorbed formulation, in which trajectories are frozen upon first contact with the boundary, provides the rigorous framework for extinction-permitting dynamics. 
This formulation addresses questions about extinction probabilities, mean persistence times, and quasi-stationary distributions, reflecting the irreversibility of CTMC extinction. 
These are not competing models for the same quantity; they answer fundamentally different biological questions. We argue that their distinction must be stated explicitly to avoid both modeling ambiguity and overreach in well-posedness claims.

To rigorously support the open-domain formulation, we develop a two-stage well-posedness architecture tailored to the boundary-degenerate structure of demographic diffusions. 
The first stage establishes a maximal strong solution together with the absence of interior blow-up before boundary hitting, strictly without assuming positivity invariance. 
The second stage demonstrates that positivity invariance (i.e., the boundary is not hit in finite time almost surely) only follows under an additional, verifiable boundary-barrier Lyapunov condition. 
Separating these two stages is both mathematically precise and biologically meaningful. 
It prevents the tacit assumption of absolute survival from being embedded into a baseline well-posedness statement and clarifies that the absorbed and open-domain formulations differ not in their local SDE specification but in their global boundary treatment.

Ultimately, this paper provides a definitive template for boundary-aware, covariance-consistent stochastic modeling.
The current work is structured to build this framework systematically, and the main contributions of this paper are as follows:
\begin{enumerate}[label=(C\arabic*)]
\item Mechanistic derivation of a full-covariance demographic diffusion.
Starting from an integer-valued CTMC with Bernoulli-coupled predation--conversion events, we rigorously derive the CME and its associated CLE. 
This establishes that the diffusion covariance,
\[
a(x)=\frac{1}{\Omega}\,S\,\mathrm{diag}(f(x))\,S^\top,
\]
is strictly dictated by event stoichiometry and channel intensities, replacing phenomenological noise constructs with a purely mechanistic foundation.
\item Identification of a structurally negative predator--prey cross-covariance.
We mathematically isolate the effect of coupled predation--conversion events, proving they inherently generate a strictly negative cross-covariance, $a_{12}(x) < 0$, on the interior domain $D = (0,\infty)^2$. 
We demonstrate that drift-matched split-channel closures fail to capture this geometric feature. This result formalizes the limitations of diagonal-noise approximations and identifies regimes where they do not accurately represent the local fluctuation structure.
\item Formalization of the open-domain versus absorbed modeling bifurcation.
We address a persistent ambiguity in stochastic ecological modeling by explicitly separating two global diffusion formulations built from the same local CME-consistent coefficients. 
We distinguish an open-domain formulation tailored for interior, survival-conditioned coexistence dynamics from an absorbed formulation designed for extinction-permitting dynamics, ensuring the mathematical framework aligns with the targeted biological observables.
\item A two-stage well-posedness architecture for boundary-degenerate diffusions.
We develop a tailored analytical framework for the open-domain formulation that separates non-explosion from boundary avoidance. 
First, we establish maximal strong solvability without interior blow-up before boundary contact via a Lyapunov anti-explosion condition. 
Second, we establish positivity invariance as a distinct property dependent on an additional boundary-barrier condition, preventing the tacit assumption of survival from being embedded in baseline well-posedness claims.
\end{enumerate}

Figure~\ref{fig:mechanism_summary_tikz} presents the mechanistic roadmap and modeling bifurcation for the R--M model. The structure of this work is as follows.
Section~\ref{sec:literature_paperA} reviews the relevant literature and positions the present work relative to existing stochastic predator--prey frameworks, CME/CLE derivations, and demographic diffusion well-posedness results. 
Section~\ref{sec:deterministic_backbone} records the deterministic R--M backbone. Sections~\ref{sec:ctmc_cme_paperA}--\ref{sec:diffusion_covariance_paperA} develop the CTMC/CME formulation and the full-covariance diffusion derivation. 
Section~\ref{sec:formulations_paperA} formalizes the open-domain and absorbed formulation distinction. Section~\ref{sec:wellposed_paperA} develops the two-stage well-posedness framework. 
Section~\ref{sec:numerical_simulation} presents the numerical simulations that verify the boundary-aware, covariance-consistent modeling structure.
Section~\ref{sec:discussion_paperA} discusses implications and scope. Appendices provide full derivations of the CME-to-diffusion reduction and complete proofs of the well-posedness theorems.

\section{Literature review and positioning}
\label{sec:literature_paperA}

The present work connects several partially overlapping bodies of literature: 
(i) stochastic predator--prey modeling and its mechanistic foundations, 
(ii) demographic-noise amplified population oscillation,
(iii) CME and chemical Langevin methods in population biology,
and (iv) well-posedness and boundary analysis for diffusion processes arising in ecology and mathematical biology. 
We review each part in turn and identify the specific gaps that motivate the contributions of this paper.

\subsection{Stochastic predator--prey models: noise structures and their origins}

The introduction of noise into predator--prey systems has a long history. 
Early works by May \citep{may_stability_2001} and Gard \citep{gard_persistence_1984} studied stochastic Lotka--Volterra-type equations with environmental noise. 
They establish that random fluctuations can qualitatively alter or preserve persistence and coexistence depending on noise intensity. 
Subsequently, a large literature developed stochastic R--M models and related systems, using both It\^o and Stratonovich SDE formulations with various noise structures \citep{yuan_bifurcation_2023,barraquand_no_2023,braumann_variable_2002,braumann_ito_2007,liu_survival_2011,zhao_survival_2015}.

A common modeling approach in this literature is to introduce multiplicative noise directly at the SDE level, typically by adding separate, independent Wiener processes to the prey and predator equations:
\begin{equation*}
dN = f(N,P)\,dt + \sigma_1 N\,dW_1,\qquad
dP = g(N,P)\,dt + \sigma_2 P\,dW_2,
\end{equation*}
or variants thereof \citep{huang_stochastic_2021,du_conditions_2016,qi_threshold_2021,yasin_spatio-temporal_2023,liu_dynamics_2018}. 
Such constructions are mathematically tractable and allow for clear persistence and extinction analyses.
However, they remain fundamentally phenomenological rather than mechanistic because the noise intensities $\sigma_1$ and $\sigma_2$ are free parameters not derivable from event stoichiometry, and the independence of the two noise sources is imposed rather than justified. 
In particular, a predation event that simultaneously removes a prey individual and possibly adds a predator individual counts as a single demographic event. 
Its natural diffusion-level representation must involve a coupled, non-diagonal fluctuation increment rather than two independent ones. 
This mechanistic structure is inherently absent from diagonal noise models, and quantifying the severe consequences of this omission is a primary motivation for the present work.

A different strand of stochastic ecological modeling introduces noise through environmental (extrinsic) stochasticity rather than demographic (intrinsic) stochasticity \citep{wu_stochastic_2019,chatterjee_predatorprey_2024,das_modelling_2021,belabbas_rich_2021}. 
Environmental noise models capture variation in parameters (e.g., carrying capacity or predation rate) driven by external fluctuations, and they also typically produce diagonal or block-diagonal noise structures. 
This paper focuses exclusively on demographic noise, arising from the discreteness of individual events in finite populations, and does not incorporate environmental noise. 
This distinction is standard in the theory of stochastic population models \citep{newman_extinction_2004,ovaskainen_stochastic_2010,wang_analysis_2025,dobramysl_environmental_2013} and is strictly maintained throughout.

Several authors have emphasized the importance of the mechanistic derivation of noise structure in population models. 
Renshaw \citep{renshaw_modelling_1995} and Allen \citep{allen_introduction_2010} discuss CTMC formulations of birth--death--immigration processes and their diffusion limits. 
Naghnaeian and Del Vecchio \citep{naghnaeian_robust_2017} and Wang \citep{wang_analysis_2025-1} develop moment approximations for interacting population models starting from the CME. 
Recently, Black and McKane \citep{black_stochastic_2012} provide an accessible account of CME-based stochastic modeling in ecology, while Rogers et al.\ \citep{rogers_demographic_2012} use CME methods to study demographic fluctuations in multi-species systems. 
In these works, the importance of deriving the diffusion covariance from event stoichiometry rather than imposing it phenomenologically is recognized.
However, the specific structural consequences for predator--prey covariance geometry and their implications for boundary behavior have not been systematically unraveled.

\subsection{Quasi-cycles, stochastic amplification, and resonance near Hopf bifurcations}

The phenomenon of quasi-cycles in predator--prey and other ecological systems, which are stochastic oscillations in the stable regime driven by the amplification of demographic fluctuations, was clarified theoretically by McKane and Newman \citep{mckane_predator-prey_2005}.
They used an individual-level model for a predator--prey system to demonstrate this behavior. 
Their analysis showed that demographic noise can produce power spectral peaks at a characteristic frequency, even when the deterministic equilibrium is stable.
This result provides a mechanistic explanation for observed population oscillations without invoking nonlinear limit cycles. 
This result was extended to R--M-type systems by Barraquand \citep{barraquand_no_2023}, to spatial models by Butler and Goldenfeld \citep{butler_robust_2009,butler_fluctuation-driven_2011}, and to systems with multiple interacting species by Biancalani et al.\ \citep{biancalani_noise-induced_2012}. Theoretical tools designed to capture these phenomena, such as the linear noise approximation (LNA), rely explicitly on the evaluation of the diffusion covariance matrix at equilibrium, $D_* = a(K_3)$, to quantify these spectral peaks. However, without a bottom-up derivation, $D_*$ is often misspecified as diagonal.

On the other hand, stochastic sensitivity analysis for nonlinear stochastic systems has been developed by Bashkirtseva and colleagues \citep{bashkirtseva_stochastic_2004,bashkirtseva_stochastic_2014,bashkirtseva_stochastic_2018}, including applications to mathematical ecology models. 
Their stochastic sensitivity function (SSF) formalism provides a geometric description of fluctuation amplification via Lyapunov-equation-based covariance analysis, and quantifies the local propensity for noise-induced transitions.
Consequently, noisy precursors and early-warning signals near tipping points and bifurcations have attracted substantial interest as indicators of regime shifts in ecological systems \citep{scheffer_early-warning_2009,dakos_slowing_2008,boettiger_quantifying_2012,carpenter_early_2011,lenton_early_2011,boettiger_no_2013}. 
Our rigorous identification of the negative cross-covariance provides the exact geometric input needed for these sensitivity and early-warning frameworks.
This examines whether risk assessments of noise-induced extinction is fundamentally biased by diagonal-noise assumptions.

\subsection{CME and Langevin methods in population biology}

The CME and the associated CLE derived via van Kampen system-size expansions are classical \citep{van_kampen_stochastic_2007,gardiner_stochastic_2009,gillespie_chemical_2000}. 
Gillespie \citep{gillespie_chemical_2000} gives a rigorous derivation of the CLE from the CME under a ``many firings'' condition, yielding a diffusion approximation whose covariance is the density-rescaled version of the exact CTMC infinitesimal covariance. 
The key formula
\[
a(x)=\frac{1}{\Omega}\,S\,\mathrm{diag}(f(x))\,S^\top
\]
appears in this context as a direct mathematical consequence of the event stoichiometry and propensities; see also Kampen \citep{van_kampen_stochastic_2007} and Elf and Ehrenberg \citep{elf_fast_2003}.

Applications of the CME/CLE framework to ecological population models have been developed by Allen \citep{allen_introduction_2010}, Black and McKane \citep{black_stochastic_2012}, Ovaskainen and Meerson \citep{ovaskainen_stochastic_2010}, and others. 
In predator--prey contexts specifically, Pineda-Krch et al.\ \citep{pineda-krch_gillespiessa_2009} apply Gillespie SSA to predator--prey systems; Alonso et al.\ \citep{alonso_stochastic_2007} study demographic stochasticity and critical transitions using CME-based approximations. 
While the present paper builds upon this robust CLE/CME framework, it pivots from general moment closures to focus specifically on: 
(i) the mechanistic isolation of the non-diagonal covariance structure, 
(ii) the explicit comparison between coupled and split-channel CTMC closures, 
and (iii) the rigorous treatment of boundary behavior in the resulting diffusion.

\subsection{Well-posedness and boundary analysis for population diffusions}

The well-posedness of SDEs arising in mathematical biology has been studied extensively. This is especially true for models with boundary-degenerate or state-constrained diffusion.  
For one-dimensional diffusions, Feller's boundary classification \citep{feller_parabolic_1952,feller_diffusion_1954} provides explicit conditions for the accessibility and attainability of boundaries in terms of drift--diffusion balance. 
Applications to population genetics (Wright--Fisher diffusion) are classical \citep{ethier_markov_1986,hofrichter_diffusion_2014,ewens_mathematical_2004}. 
For square-root diffusions (CIR processes) \citep{cox_theory_1985}, the Feller condition $2\kappa\theta\ge\sigma^2$ governs whether the boundary is accessible, and is widely used in financial mathematics and population modeling.

In two or higher dimensions, direct analogues of the Feller criterion are generally unavailable because boundary attainability depends on the full nonlinear drift--diffusion interaction and may differ across distinct boundary faces. 
Lyapunov function methods provide the standard replacement: non-explosion and positivity are established by finding a Lyapunov function satisfying a linear growth bound of the generator \citep{khasminskii_stochastic_2012,meyn_stability_1993,mao_stochastic_2008}. 
These methods have been used for stochastic Lotka--Volterra systems and related models \citep{cui_long-term_2026,zhao_stochastic_2024,zhan_dynamical_2024,liu_asymptotic_2020}, typically in the context of global well-posedness and long-time behavior under various noise conditions. A relevant model is the mathematical Leslie–Gower system, where a stochastic extension driven by demographic noise is established \citep{wang_algebraicspectral_2026}, but the well-posedness analysis is missing.  

Our work adopts this Lyapunov perspective for the boundary-degenerate demographic diffusion arising from the R--M CTMC, but introduces a critical structural departure from the existing literature: a deliberate, two-stage separation of well-posedness claims.
We first establish a maximal strong solution with no interior blow-up before boundary hitting (Theorem~\ref{thm:maximal_no_interior_explosion_paperA}), and only then establish positivity invariance under an additional barrier condition (Proposition~\ref{prop:barrier_invariance_paperA}). 
This separation is explicitly motivated by the modeling bifurcation between open-domain and absorbed formulations.
It also prevents positivity invariance, which is biologically relevant only in the open-domain formulation, from being inadvertently embedded as a baseline mathematical assumption in extinction-permitting modeling.

\subsection{Gaps and positioning of this paper}

In summary, the present paper addresses three specific gaps at the intersection of the literature reviewed above.

\paragraph{Gap 1: Mechanistic non-diagonal covariance in stochastic predator--prey diffusions.}
The CLE/CME framework is well established in the theoretical literature, and the covariance identity $a(x) = \frac{1}{\Omega}S\,\mathrm{diag}(f(x))S^\top$ is mathematically known. However, the specific structural implications for predator--prey models have rarely been isolated and analyzed as a primary modeling object.
In particular, the structural negativity of the predator--prey cross-covariance induced by coupled predation--conversion has seldom been examined in detail.
In most applied stochastic R--M works, the noise structure is specified directly at the SDE level without reference to event stoichiometry, obscuring the distinction between coupled and split-channel closures. The present paper systematically closes this gap.

\paragraph{Gap 2: Formulation distinction between open-domain and absorbed diffusions.}
The distinction between interior, survival-conditioned dynamics and extinction-permitting dynamics is heavily implicit in the ecological modeling literature, but is rarely stated as an explicit modeling bifurcation at the fundamental SDE level. 
Well-posedness results for stochastic population diffusions typically either prove positivity invariance, which implicitly enforces an open-domain formulation \citep{cui_long-term_2026,zhao_stochastic_2024,zhan_dynamical_2024,liu_asymptotic_2020},
or study the absorbed/exit-time problem directly \citep{nasell_extinction_2001,cattiaux_quasi-stationary_2009}.
These studies generally do not connect the two approaches as alternative global formulations built from the same local coefficients. 
We formalize this distinction, demonstrating its profound implications for rigorous mathematical analysis and proper ecological interpretation.

\paragraph{Gap 3: Two-stage well-posedness architecture separating non-explosion from positivity invariance.}
Standard well-posedness results for stochastic ecological diffusions typically either prove global well-posedness (including positivity invariance) in a single step  \citep{cui_long-term_2026,zhao_stochastic_2024,zhan_dynamical_2024,liu_asymptotic_2020}, or study local solutions without explicitly discussing positivity. 
The deliberate two-stage separation is a novel architectural choice.
First, non-explosion is established (Theorem~\ref{thm:maximal_no_interior_explosion_paperA}), 
and then positivity invariance is isolated under a distinct barrier condition (Proposition~\ref{prop:barrier_invariance_paperA}).
This approach is motivated by the need to maintain the open-domain/absorbed modeling bifurcation and prevent mathematical overclaiming. 
We elevate this separation into a primary organizing principle for multidimensional demographic diffusions.

\section{Deterministic R--M backbone}
\label{sec:deterministic_backbone}

This section records the deterministic R--M structure used as the background for the stochastic construction developed later. 

We consider the nondimensional R--M predator--prey system with Holling type~II predation \citep{grunert_evolutionarily_2021}:
\begin{equation}
\label{eq:RM_ODE_paperA}
\begin{cases}
\dfrac{dN}{dt} = N\!\left(1-\dfrac{N}{k}\right) - \dfrac{mNP}{1+N}, \\[8pt]
\dfrac{dP}{dt} = P\!\left(-c + \dfrac{mN}{1+N}\right),
\end{cases}
\qquad (N(t),P(t))\in Q:=\{(N,P):N\ge 0,\ P\ge 0\},
\end{equation}
where \(N\) and \(P\) denote prey and predator densities, respectively, and \(k>0\), \(m>0\), \(c>0\) are the (scaled) carrying-capacity, predation/assimilation, and predator-mortality parameters. The biologically relevant closed state space is the positive quadrant \(Q\), while the open quadrant
\begin{equation}
\label{eq:positive_quadrant_paperA}
D:=(0,\infty)^2
\end{equation}
will be the ambient domain for the open-domain diffusion formulation in Sections~\ref{sec:formulations_paperA}--\ref{sec:wellposed_paperA}. The vector field in \eqref{eq:RM_ODE_paperA} is smooth on \(Q\) (indeed on \(\{N>-1,\ P\in \mathbb R\}\)).

We work with the standard reduced nondimensional R--M form in which the predator conversion factor is not written explicitly in the $e\frac{mNP}{1+N}$ term in the $P$ equation in \eqref{eq:RM_ODE_paperA}. In the mechanistic stochastic derivation (Section~\ref{sec:ctmc_cme_paperA}), this corresponds to the normalization \(e=1\) when matching the reduced deterministic parameterization used throughout the analysis.

System~\eqref{eq:RM_ODE_paperA} has three equilibria:
\begin{align}
K_1 &= (0,0), \label{eq:K1_paperA}\\
K_2 &= (k,0), \label{eq:K2_paperA}\\
K_3 &= (N^*,P^*) 
= \left(\frac{c}{m-c},\ \frac{k(m-c)-c}{k(m-c)^2}\right). \label{eq:K3_paperA}
\end{align}
The coexistence equilibrium \(K_3\) lies in \(D\) if and only if
\begin{equation}
\label{eq:coexistence_feasibility_paperA}
m>c
\qquad \text{and} \qquad
k(m-c)>c.
\end{equation}
The first condition ensures that the predator per-capita growth term
\[
-c+\frac{mN}{1+N}
\]
can become nonnegative for sufficiently large prey density, while the second condition is exactly the positivity condition \(P^*>0\). In the present paper, \(K_3\) serves as the deterministic coexistence reference state for the mechanistic stochastic modeling and the boundary-aware diffusion formulations.

Let
\[
f(N,P)=N\!\left(1-\frac{N}{k}\right)-\frac{mNP}{1+N},
\qquad
g(N,P)=P\!\left(-c+\frac{mN}{1+N}\right).
\]
The Jacobian matrix is
\begin{equation}
\label{eq:Jacobian_paperA}
J(N,P)=
\begin{pmatrix}
1-\dfrac{2N}{k}-\dfrac{mP}{(1+N)^2} & -\dfrac{mN}{1+N}\\[10pt]
\dfrac{mP}{(1+N)^2} & -c+\dfrac{mN}{1+N}
\end{pmatrix}.
\end{equation}
At the coexistence equilibrium \(K_3\), one has \(\det J(K_3)>0\) whenever \(K_3\in D\), so local stability is determined by the trace. A direct calculation gives the classical enrichment-driven Hopf threshold \citep{grunert_evolutionarily_2021}
\begin{equation}
\label{eq:Hopf_threshold_paperA}
k_H=\frac{m+c}{m-c}.
\end{equation}
Hence \(K_3\) is locally asymptotically stable for \(k<k_H\) and unstable for \(k>k_H\), with a Hopf bifurcation at \(k=k_H\). We primarily use \eqref{eq:Hopf_threshold_paperA} to organize parameter regimes
\begin{equation}
\label{eq:Lambda_regions_paperA}
\Lambda_2 \coloneqq \{(m,c,k): k<k_H\},
\qquad
\Lambda_1 \coloneqq \{(m,c,k): k>k_H\}.
\end{equation}
Figure~\ref{fig:deterministic_backbone} illustrates the bifurcation diagram in the carrying capacity, and the phase portraits for the two regimes $\Lambda_2$ (stable coexistence) and $\Lambda_1$ (unstable oscillation), showing the Hopf bifurcation.

\begin{figure}[htbp]
    \centering
    \includegraphics[width=\textwidth]{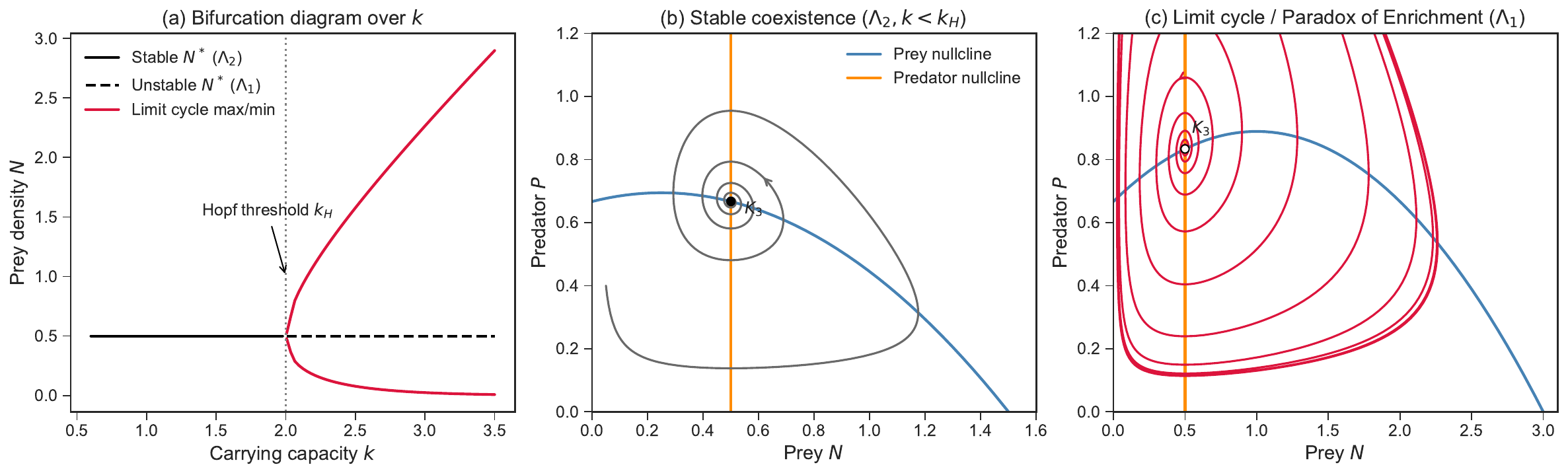}
    \caption{\textbf{Deterministic R--M backbone and the Hopf bifurcation.} 
    (a) Bifurcation diagram over the carrying capacity $k$. The coexistence equilibrium $K_3$ loses stability at the Hopf threshold $k_H$, partitioning the parameter space into the stable regime $\Lambda_2$ and the oscillatory regime $\Lambda_1$. Solid curves denote the extrema of the stable limit cycle, and dashed lines denote the unstable equilibrium. 
    (b) Phase portrait in the stable regime $\Lambda_2$ ($k < k_H$). The predator nullcline (vertical) intersects the prey nullcline (humped) to the right of its peak, resulting in trajectories spiraling inward to a stable coexistence state $K_3$. 
    (c) Phase portrait in the oscillatory regime $\Lambda_1$ ($k > k_H$). Enrichment shifts the intersection to the left of the peak, destabilizing $K_3$ and generating a stable limit cycle that drives large-amplitude, low-density excursions.}
    \label{fig:deterministic_backbone}
\end{figure}

\section{Mechanistic CTMC and CME}
\label{sec:ctmc_cme_paperA}

We now formulate an integer-valued CTMC for the R--M predator--prey system and record the associated CME. The purpose of this section is to make the event mechanism explicit at the count level and to identify the exact infinitesimal covariance in count variables. This count-scale covariance is the object that will determine the diffusion covariance in Section~\ref{sec:diffusion_covariance_paperA} after density rescaling.

Let
\[
X(t)=\bigl(X_N(t),X_P(t)\bigr)^\top \in \mathbb N_0^2
\]
denote prey and predator counts in a well-mixed population of system size \(\Omega\) (e.g.,\ habitat volume or scaling parameter), and define the corresponding density process
\[
x(t)=\frac{X(t)}{\Omega}=\bigl(N(t),P(t)\bigr)^\top\in\mathbb R_+^2.
\]
Because the exact CTMC is integer-valued, all jump increments in \(X\) must lie in \(\mathbb Z^2\). In particular, a compact predation--conversion increment of the form \(({-1},e)^\top\) with \(e\in(0,1]\) is not an admissible exact CTMC jump unless \(e\in\mathbb N\). We therefore formulate the CTMC using integer-valued channels (with Bernoulli conversion), and only later use an effective \(({-1},e)^\top\) increment as a density-level diffusion closure in Section~\ref{sec:diffusion_covariance_paperA}. This distinction is essential for mechanistic consistency.

We use the following five integer-valued event channels:
\begin{align*}
\mathcal R_1:\ & X_N \longrightarrow X_N+1 
&&\text{(prey birth)},\\
\mathcal R_2:\ & X_N \longrightarrow X_N-1 
&&\text{(prey competition death)},\\
\mathcal R_3:\ & X_P \longrightarrow X_P-1 
&&\text{(predator death)},\\
\mathcal R_4^{(B)}:\ & (X_N,X_P) \longrightarrow (X_N-1,\;X_P+1)
&&\text{(predation with successful conversion)},\\
\mathcal R_5^{(B)}:\ & (X_N,X_P) \longrightarrow (X_N-1,\;X_P)
&&\text{(predation without conversion)}.
\end{align*}
This is equivalent to a predation encounter process followed by Bernoulli conversion with success probability \(e\in(0,1]\). The stoichiometric increments are
\begin{equation}
\label{eq:nu_Bernoulli_paperA}
\nu_1=\binom{1}{0},\qquad
\nu_2=\binom{-1}{0},\qquad
\nu_3=\binom{0}{-1},\qquad
\nu_4^{(B)}=\binom{-1}{1},\qquad
\nu_5^{(B)}=\binom{-1}{0},
\end{equation}
and the corresponding stoichiometric matrix is
\begin{equation}
\label{eq:Stoichiometry_B_paperA}
S_B:=\big[\nu_1\ \nu_2\ \nu_3\ \nu_4^{(B)}\ \nu_5^{(B)}\big]\in\mathbb R^{2\times 5}.
\end{equation}

We adopt density-dependent (Kurtz-type) propensities
\begin{equation}
\label{eq:propensity_density_dependent_paperA}
\lambda_k(X)=\Omega\,f_k(x),\qquad x=\frac{X}{\Omega},
\end{equation}
with intensity functions
\begin{align}
f_1(x)&=N, &
f_2(x)&=\frac{N^2}{k}, &
f_3(x)&=cP, \label{eq:fk_basic_paperA}\\
f_4^{(B)}(x)&=e\,\frac{mNP}{1+N}, &
f_5^{(B)}(x)&=(1-e)\,\frac{mNP}{1+N}. \label{eq:fk_Bernoulli_pred_paperA}
\end{align}
Let \(f_B=(f_1,f_2,f_3,f_4^{(B)},f_5^{(B)})^\top\). Then the deterministic drift induced by the CTMC is
\begin{equation}
\label{eq:drift_Bernoulli_CTMC_paperA}
b_B(x)=S_B f_B(x)=
\begin{pmatrix}
N\!\left(1-\dfrac{N}{k}\right)-\dfrac{mNP}{1+N}\\[6pt]
-cP+e\,\dfrac{mNP}{1+N}
\end{pmatrix}.
\end{equation}
This is the R--M drift with explicit conversion efficiency \(e\). 

In the current mechanistic CTMC/CME formulation, \(e\in(0,1]\) enters the predator gain term but not the prey loss term. 
To maintain exact consistency with the deterministic backbone~\eqref{eq:RM_ODE_paperA} used in Section~\ref{sec:deterministic_backbone}, we perform the main stochastic analysis under the normalization
\[
e=1.
\]
The mechanistic derivations in this section are written for general \(e\) where useful (to expose the origin of covariance structure), and the specialization \(e=1\) is imposed when matching the reduced R--M parameterization used in Sections~\ref{sec:formulations_paperA}--\ref{sec:numerical_simulation}.

For later comparison, it is useful to note that one can construct an alternative integer-valued CTMC with separate channels for prey removal and predator reproduction, using increments \(({-1},0)^\top\) and \((0,1)^\top\) with rates proportional to \(\frac{mNP}{1+N}\) and \(e\,\frac{mNP}{1+N}\), respectively. This split-channel construction reproduces the same deterministic drift \eqref{eq:drift_Bernoulli_CTMC_paperA} but changes the count-level and diffusion-level covariance structure. In particular, it removes the same-event predator--prey coupling in the predation contribution and will therefore serve as a diagonal-noise comparator in Section~\ref{sec:diffusion_covariance_paperA}, not as the primary mechanistic baseline.

We formulate the CME.
Let \(p(X,t)=\mathbb P(X(t)=X)\), \(X\in\mathbb N_0^2\). For a finite family of channels with increments \(\nu_k\) and propensities \(\lambda_k(X)\), the CME is
\begin{equation}
\label{eq:CME_general_paperA}
\frac{\partial}{\partial t}p(X,t)
=
\sum_{k=1}^{K}
\Big[
\lambda_k(X-\nu_k)\,p(X-\nu_k,t)
-
\lambda_k(X)\,p(X,t)
\Big].
\end{equation}
In the present Bernoulli-coupled model, \(K=5\), the increments are those in \eqref{eq:nu_Bernoulli_paperA}, and the propensities are given by \eqref{eq:propensity_density_dependent_paperA}--\eqref{eq:fk_Bernoulli_pred_paperA}. We use \eqref{eq:CME_general_paperA} primarily as the exact count-level probabilistic description from which moment identities and the diffusion covariance inherit their structure.

The associated CTMC generator acting on test functions \(\varphi:\mathbb N_0^2\to\mathbb R\) is
\begin{equation}
\label{eq:CTMC_generator_paperA}
(\mathcal A\varphi)(X)
=
\sum_{k=1}^{K}\lambda_k(X)\,\big(\varphi(X+\nu_k)-\varphi(X)\big).
\end{equation}
As recalled in Appendix~\ref{app:CME_diffusion}, this yields both the count-scale drift increment
\[
\sum_{k=1}^K \lambda_k(X)\,\nu_k
\]
and the exact infinitesimal covariance (quadratic variation density) in count variables:
\begin{equation}
\label{eq:inf_cov_counts_paperA}
G(X)=\sum_{k=1}^{K}\lambda_k(X)\,\nu_k\nu_k^\top.
\end{equation}
Equivalently, in stoichiometric form,
\begin{equation}
\label{eq:inf_cov_counts_stoich_paperA}
G(X)=S\,\mathrm{diag}\!\big(\lambda(X)\big)\,S^\top,
\end{equation}
for the chosen channel representation \((S,\lambda)\). This identity is exact at the CTMC level and already shows that second-order structure depends on event stoichiometry, not only on the drift.

The key implication of \eqref{eq:inf_cov_counts_paperA}--\eqref{eq:inf_cov_counts_stoich_paperA} is that, under density scaling \(x=X/\Omega\) with \(\lambda_k(X)=\Omega f_k(x)\), the diffusion covariance in the demographic-noise approximation is inherited directly from the count-level event structure. In Section~\ref{sec:diffusion_covariance_paperA}, we make this explicit through the unified formula
\[
a(x)=\frac{1}{\Omega}\,S\,\mathrm{diag}(f(x))\,S^\top,
\]
and use it to compare Bernoulli-coupled, effective coupled \(({-1},e)^\top\), and split-channel closures. This is where the structural negative predator--prey cross-covariance appears transparently.

\section{CME-consistent diffusion approximation and covariance structure}
\label{sec:diffusion_covariance_paperA}

This section derives the demographic diffusion approximation from the CTMC/CME formulation in Section~\ref{sec:ctmc_cme_paperA} and isolates the covariance structure induced by event stoichiometry. The central point is that the diffusion covariance is not a free modeling choice once the event channels are specified: it is inherited from the exact count-scale infinitesimal covariance. In particular, coupled predation--conversion generates a structurally negative predator--prey cross-covariance term, whereas a drift-matched split-channel construction does not.

Let \(S\) denote a chosen stoichiometric matrix and \(f(x)\) the associated vector of density-level intensities. Under the density-dependent scaling \(\lambda_k(X)=\Omega f_k(X/\Omega)\), the chemical Langevin diffusion approximation for the density process \(x_t=X_t/\Omega\) takes the It\^o form
\begin{equation}
\label{eq:CLE_density_paperA}
dx_t = b(x_t)\,dt + \Sigma(x_t)\,dW_t,
\end{equation}
with drift
\begin{equation}
\label{eq:drift_unified_paperA}
b(x)=S f(x),
\end{equation}
and diffusion covariance
\begin{equation}
\label{eq:covariance_unified_paperA}
a(x)\coloneqq \Sigma(x)\Sigma(x)^\top
=
\frac{1}{\Omega}\,S\,\mathrm{diag}\!\bigl(f(x)\bigr)\,S^\top.
\end{equation}
Equation~\eqref{eq:covariance_unified_paperA} is the density-scale counterpart of the exact CTMC infinitesimal covariance identity in counts (Section~\ref{sec:ctmc_cme_paperA}); see Appendix~\ref{app:CME_diffusion} for the derivation.

To isolate the predation contribution across different closures, we write
\begin{equation}
\label{eq:fpred_paperA}
f_{\mathrm{pred}}(x):=\frac{mNP}{1+N},
\qquad x=(N,P)\in D=(0,\infty)^2.
\end{equation}
Below we compare three drift-compatible predation closures (up to parameterization): (i) the exact Bernoulli-coupled CTMC mechanism, (ii) an effective diffusion-level coupled channel \(({-1},e)^\top\), and (iii) a split-channel comparator. Their difference lies entirely in the induced second-order structure.

For the Bernoulli-coupled predation channels \(\mathcal R_4^{(B)}\) and \(\mathcal R_5^{(B)}\) from Section~\ref{sec:ctmc_cme_paperA}, the density-level predation covariance contribution is
\begin{align}
a_{\mathrm{pred}}^{(B)}(x)
&=
\frac{1}{\Omega}\Big[
f_4^{(B)}(x)\,\nu_4^{(B)}(\nu_4^{(B)})^\top
+
f_5^{(B)}(x)\,\nu_5^{(B)}(\nu_5^{(B)})^\top
\Big] \notag\\
&=
\frac{1}{\Omega}\,f_{\mathrm{pred}}(x)
\begin{pmatrix}
1 & -e\\
-e & e
\end{pmatrix}.
\label{eq:apred_B_paperA}
\end{align}
Hence, the off-diagonal predation term is
\begin{equation}
\label{eq:cross_B_negative_paperA}
\bigl(a_{\mathrm{pred}}^{(B)}\bigr)_{12}(x)
=
-\frac{1}{\Omega}\,e\,f_{\mathrm{pred}}(x)
=
-\frac{1}{\Omega}\,e\,\frac{mNP}{1+N}.
\end{equation}
In particular, for \(x\in D\) and \(e>0\), this term is strictly negative. Biologically, it records the same-event negative correlation created by prey removal and predator gain in a predation encounter with successful conversion.

For compact diffusion notation, one may use a single effective predation channel at the density level:
\begin{equation}
\label{eq:nu_eff_paperA}
\nu_{\mathrm{eff}}=\binom{-1}{e},
\qquad
f_{\mathrm{eff}}(x)=f_{\mathrm{pred}}(x).
\end{equation}
This is not an exact integer-valued CTMC jump when \(e\notin\mathbb N\), but it is a valid diffusion-level closure consistent with the drift. The corresponding predation covariance contribution is
\begin{equation}
\label{eq:apred_eff_paperA}
a_{\mathrm{pred}}^{(\mathrm{eff})}(x)
=
\frac{1}{\Omega}\,f_{\mathrm{pred}}(x)\,\nu_{\mathrm{eff}}\nu_{\mathrm{eff}}^\top
=
\frac{1}{\Omega}\,f_{\mathrm{pred}}(x)
\begin{pmatrix}
1 & -e\\
-e & e^2
\end{pmatrix}.
\end{equation}
Therefore,
\begin{equation}
\label{eq:cross_eff_negative_paperA}
\bigl(a_{\mathrm{pred}}^{(\mathrm{eff})}\bigr)_{12}(x)
=
-\frac{1}{\Omega}\,e\,f_{\mathrm{pred}}(x)
=
-\frac{1}{\Omega}\,e\,\frac{mNP}{1+N}<0
\qquad (x\in D,\ e>0).
\end{equation}
Thus, the effective coupled closure preserves the same predation-induced cross-covariance sign and magnitude as the exact Bernoulli-coupled mechanism.

Consider the split-channel predation representation with
\[
\nu_4^{(S)}=\binom{-1}{0},
\qquad
\nu_5^{(S)}=\binom{0}{1},
\qquad
f_4^{(S)}(x)=f_{\mathrm{pred}}(x),
\qquad
f_5^{(S)}(x)=e\,f_{\mathrm{pred}}(x).
\]
Then the predation-related covariance contribution is
\begin{equation}
\label{eq:apred_split_paperA}
a_{\mathrm{pred}}^{(S)}(x)
=
\frac{1}{\Omega}\Big[
f_4^{(S)}(x)\,\nu_4^{(S)}(\nu_4^{(S)})^\top
+
f_5^{(S)}(x)\,\nu_5^{(S)}(\nu_5^{(S)})^\top
\Big]
=
\frac{1}{\Omega}\,f_{\mathrm{pred}}(x)
\begin{pmatrix}
1 & 0\\
0 & e
\end{pmatrix}.
\end{equation}
Hence
\begin{equation}
\label{eq:cross_split_zero_paperA}
\bigl(a_{\mathrm{pred}}^{(S)}\bigr)_{12}(x)=0.
\end{equation}
This representation is drift-compatible with the coupled mechanisms but removes the same-event predator--prey fluctuation coupling. For that reason, we treat it as a diagonal-noise comparator, not as the default mechanistic closure.

\begin{proposition}[Structural predation cross-covariance under coupled versus split closures]
\label{prop:structural_cross_cov_paperA}
Let \(x=(N,P)\in D=(0,\infty)^2\) and \(e\in(0,1]\). For the predation intensity \(f_{\mathrm{pred}}(x)=\frac{mNP}{1+N}\), the predation contribution to the diffusion covariance satisfies:
\begin{enumerate}[label=(\roman*)]
\item under the exact Bernoulli-coupled closure,
\[
\bigl(a_{\mathrm{pred}}^{(B)}\bigr)_{12}(x)
=
-\frac{1}{\Omega}\,e\,f_{\mathrm{pred}}(x)<0;
\]
\item under the effective coupled \(({-1},e)\) closure,
\[
\bigl(a_{\mathrm{pred}}^{(\mathrm{eff})}\bigr)_{12}(x)
=
-\frac{1}{\Omega}\,e\,f_{\mathrm{pred}}(x)<0;
\]
\item under the split-channel closure,
\[
\bigl(a_{\mathrm{pred}}^{(S)}\bigr)_{12}(x)=0.
\]
\end{enumerate}
\end{proposition}

Proposition~\ref{prop:structural_cross_cov_paperA} implies that the drift equivalence does not imply covariance equivalence.
Different channel representations can produce the same deterministic drift while yielding different diffusion covariances. Indeed, the drift depends on the first stoichiometric moments
\[
b(x)=\sum_k f_k(x)\,\nu_k,
\]
whereas the covariance depends on the second stoichiometric moments
\[
a(x)=\frac{1}{\Omega}\sum_k f_k(x)\,\nu_k\nu_k^\top.
\]
Thus, channel decompositions that are indistinguishable at the ODE level may differ at the diffusion level. This is the precise mechanism behind the distinction between coupled predation closures (which produce predator--prey cross-covariance) and split-channel closures (which can suppress it).

Within the class of drift-compatible predation closures considered above, the sign and presence of predator--prey cross-covariance at the diffusion level is determined by whether prey loss and predator gain occur in the same event channel. In particular, the negative off-diagonal term is not an ad hoc modeling embellishment; it is a stoichiometric consequence of coupled predation--conversion. Figure~\ref{fig:covariance_geometry} shows the microscopic state transitions with arrows indicating the stoichiometric increments, and the macroscopic covariance ellipses with negative correlation tilting the fluctuation geometry.

\begin{figure}[htbp]
        \centering
        \includegraphics[width=\textwidth]{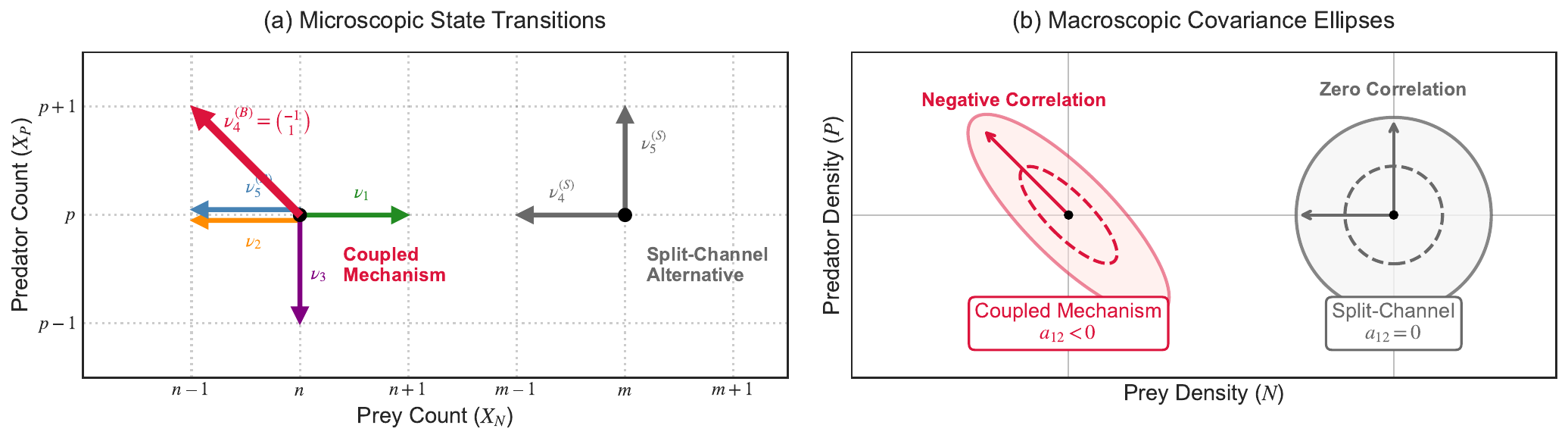}
        \caption{\textbf{Mechanistic origin of the demographic diffusion covariance.}
        (a) Microscopic state transitions on the count lattice. The exact Bernoulli-coupled mechanism strictly requires a diagonal jump vector $\nu_4^{(B)}$ (crimson) to represent a single predation--conversion event. The split-channel comparator falsely separates this into two independent, orthogonal jumps.
        (b) Macroscopic covariance geometry at the density scale. The diagonal jump in the coupled mechanism inherently tilts the local fluctuation geometry, generating a strictly negative predator--prey cross-covariance ($a_{12} < 0$). In contrast, the phenomenological split-channel closure produces an axis-aligned covariance ellipse ($a_{12} = 0$), failing to capture the true geometric structure of the demographic noise despite maintaining mean-field drift equivalence.}
        \label{fig:covariance_geometry}
\end{figure}

\begin{remark}[Bernoulli-coupled exact CTMC vs.\ effective \(({-1},e)^\top\) closure]
\label{rmk:Bernoulli_vs_effective_paperA}
The exact Bernoulli-coupled CTMC closure and the effective \(({-1},e)^\top\) diffusion closure agree in deterministic drift and in the predation-induced cross-covariance magnitude
\[
-\frac{1}{\Omega}\,e\,f_{\mathrm{pred}}(x),
\]
but they differ in the predator-variance contribution:
\[
\bigl(a_{\mathrm{pred}}^{(B)}\bigr)_{22}(x)=\frac{1}{\Omega}\,e\,f_{\mathrm{pred}}(x),
\qquad
\bigl(a_{\mathrm{pred}}^{(\mathrm{eff})}\bigr)_{22}(x)=\frac{1}{\Omega}\,e^2\,f_{\mathrm{pred}}(x).
\]
In this paper, the effective closure is used as a compact full-covariance diffusion representation, while the Bernoulli-coupled CTMC remains the exact integer-valued mechanistic foundation.
\end{remark}

For later reference, it is sometimes convenient to decompose the full covariance into non-predation and predation contributions:
\[
a(x)=a_{\mathrm{base}}(x)+a_{\mathrm{pred}}(x),
\]
where \(a_{\mathrm{base}}(x)\) collects prey birth, prey competition death, and predator death channels, and \(a_{\mathrm{pred}}(x)\) is one of
\[
a_{\mathrm{pred}}^{(B)}(x),\qquad
a_{\mathrm{pred}}^{(\mathrm{eff})}(x),\qquad
a_{\mathrm{pred}}^{(S)}(x).
\]
Since the base channels in the present construction contribute only diagonal terms, the sign and magnitude of the full off-diagonal covariance \(a_{12}(x)\) are determined entirely by the predation component. Therefore, Proposition~\ref{prop:structural_cross_cov_paperA} directly controls the sign of the full predator--prey cross-covariance in the coupled closures.

Since \(a(x)\) is symmetric positive semidefinite on \(\mathbb R_+^2\), it admits (non-unique) factorizations
\begin{equation}
\label{eq:factorization_paperA}
a(x)=\Sigma(x)\Sigma(x)^\top.
\end{equation}
Two useful choices are:
\begin{enumerate}[label=(\roman*)]
\item Event-based factorization that uses \(K\) independent Brownian drivers (one per channel):
\begin{equation*}
\Sigma_K(x)
=
\frac{1}{\sqrt{\Omega}}\,
S\,\sqrt{\mathrm{diag}\big(f(x)\big)}
\in\mathbb R^{d\times K},
\end{equation*}
which preserves channel semantics and is convenient for interpreting noise sources;
\item Minimal-dimensional square root (e.g.\ Cholesky):
choose a measurable \(2\times2\) square root \(\Sigma_2(x)\) satisfying \(\Sigma_2(x)\Sigma_2(x)^\top=a(x)\), which is convenient for numerical simulation and for the formulation of the open-domain SDE.
\end{enumerate}
These choices are equivalent in law at the level of the diffusion process defined by \eqref{eq:CLE_density_paperA}, provided they generate the same covariance matrix \(a(x)\).

\begin{assumption}[Diagonal-noise as a controlled approximation: small relative cross-covariance]
\label{ass:diag_condition_paperA}
Let \(\mathcal R\subset D\) be a state region of interest. A diagonal-noise approximation is considered only in regimes where
\begin{equation}
\label{eq:diag_condition_paperA}
\rho(x):=\frac{|a_{12}(x)|}{\sqrt{a_{11}(x)\,a_{22}(x)}}\ll 1,
\qquad x\in\mathcal R.
\end{equation}
That is, diagonalization is treated as a diagnostic approximation rather than a default mechanistic model.
\end{assumption}

\begin{remark}[Sufficient regimes for small cross-covariance]
\label{rmk:sufficient_diag_paperA}
Under the coupled closures, \(|a_{12}(x)|\) scales like \(\frac{1}{\Omega}e\,f_{\mathrm{pred}}(x)\). Thus \eqref{eq:diag_condition_paperA} may hold, for example, in any of the following regimes (or combinations thereof):
\begin{enumerate}[label=\textnormal{(\roman*)}]
\item Low conversion efficiency: \(e\ll 1\);
\item Predation-subdominant region: \(f_{\mathrm{pred}}(x)\ll f_1(x)+f_2(x)+f_3(x)\) on \(\mathcal R\);
\item State restriction away from strong interaction: \(N\) or \(P\) remains uniformly small on \(\mathcal R\), so \(f_{\mathrm{pred}}(x)=\frac{mNP}{1+N}\) is uniformly small.
\end{enumerate}
When these conditions fail, diagonal-noise closures are expected to misrepresent second-order statistics, in particular \(\mathrm{Cov}(N,P)\) and the geometry of local fluctuations.
\end{remark}

In the current work, the CTMC/CME formulation provides a microscopic description, whereas the diffusion (CLE) approximation yields a macroscopic limit. Thus, our framework establishes a principled bridge between scales, aligning conceptually with hybrid multiscale modeling approaches that couple discrete events with continuum dynamics \citep{liu_bidirectional_2025}.
With the local diffusion coefficients now fixed, the next section clarifies a distinct issue: the same local SDE coefficients can be embedded into two different global diffusion formulations on the positive quadrant, namely an open-domain formulation for interior dynamics and an absorbed formulation for extinction-permitting dynamics.

\section{Open-domain and absorbed diffusion formulations}
\label{sec:formulations_paperA}

The mechanistic construction in Sections~\ref{sec:ctmc_cme_paperA}--\ref{sec:diffusion_covariance_paperA} determines the local diffusion coefficients (drift and covariance) from event stoichiometry and channel intensities. A separate modeling choice concerns the global interpretation of the diffusion on the positive quadrant: whether one studies interior dynamics on the open coexistence domain, or extinction-permitting dynamics with boundary absorption. This distinction is central in the present paper and should be made explicit before any well-posedness statement is formulated.

Throughout, we write the density process as \(x_t=(N_t,P_t)^\top\), and use the closed positive quadrant $Q$ together with its interior $D$.
The local diffusion coefficients are taken from the CME-consistent construction:
\begin{equation}
\label{eq:common_local_coeffs_paperA}
b(x)=S f(x),\qquad
a(x)=\Sigma(x)\Sigma(x)^\top=\frac{1}{\Omega}\,S\,\mathrm{diag}(f(x))\,S^\top,
\end{equation}
with \((S,f)\) chosen consistently with the closure under study (in particular, the full-covariance coupled closure as the default mechanistic representation, and the split-channel model as a comparator). The distinction introduced in this section concerns how these same local coefficients are embedded into a global stochastic process on \(D\) or \(Q\).

The open-domain formulation is the diffusion model posed on the interior \(D\), with initial condition \(x_0\in D\):
\begin{equation}
\label{eq:open_domain_SDE_paperA}
dx_t=b(x_t)\,dt+\Sigma(x_t)\,dW_t,
\qquad x_0\in D.
\end{equation}
Under the full-covariance coupled closure, the predation mechanism induces a negative off-diagonal term \(a_{12}(x)<0\) for \(x\in D\), as shown in Section~\ref{sec:diffusion_covariance_paperA}. At this stage, \eqref{eq:open_domain_SDE_paperA} is a local SDE statement on an open set; the issue of boundary attainability is deferred to Section~\ref{sec:wellposed_paperA}.

The open-domain formulation is intended to describe interior stochastic dynamics on the coexistence state space \(D\), rather than extinction events themselves. It is the natural framework for analyzing observables such as:
(i) local fluctuation geometry and covariance structure near coexistence,
(ii) survival-conditioned dynamics over finite or moderate horizons,
and (iii) local approximations (e.g., linearization-based analyses) that presuppose the process remains in the interior. In this interpretation, one studies the diffusion as a model of demographic variability conditional on coexistence, and the main mathematical questions concern local strong solvability, non-explosion before boundary hitting, and (under additional conditions) positivity invariance.

To model extinction-permitting dynamics, one uses an absorbed formulation based on the same local coefficients. Let
\begin{equation}
\label{eq:tau_boundary_paperA}
\tau_{\partial D}:=\inf\{t\ge 0:\ x_t\notin D\}
\end{equation}
denote the first exit (boundary hitting) time from the interior for a local solution of \eqref{eq:open_domain_SDE_paperA}. The absorbed process is then defined by freezing at the first boundary hit:
\begin{equation}
\label{eq:absorbed_process_definition_paperA}
x_t^{\mathrm{abs}}:=
\begin{cases}
x_t, & t<\tau_{\partial D},\\[4pt]
x_{\tau_{\partial D}}, & t\ge \tau_{\partial D}.
\end{cases}
\end{equation}
Biologically, the boundary \(\partial D\) corresponds to loss of strict coexistence:
\(N=0\) (prey extinction), \(P=0\) (predator extinction), or \((0,0)\) (total extinction). The freezing rule in \eqref{eq:absorbed_process_definition_paperA} is the diffusion-level analogue of the irreversibility of extinction in the underlying finite-population CTMC.

The absorbed formulation is the appropriate diffusion model when the scientific question concerns extinction-related observables, such as:
extinction probabilities over a given time horizon, extinction-time distributions, or mean persistence times. In this setting, boundary hitting is not a pathology to be excluded; it is part of the model output. Thus, the absorbed formulation and the open-domain formulation are not competing descriptions of the same quantity, but rather two diffusion-level models tailored to different observables and time-horizon interpretations.

\begin{remark}[Why these are different questions (and why the distinction matters)]
The open-domain and absorbed formulations differ in at least three ways. First, they target different observables: interior fluctuation structure versus extinction/permanence statistics. Second, they differ in time-horizon emphasis: open-domain analyses are often local or survival-conditioned, whereas absorbed analyses directly encode rare-event accumulation and eventual boundary contact. Third, they correspond to different analogies with the exact CTMC: the absorbed formulation mirrors CTMC extinction irreversibility, while the open-domain formulation serves as a mathematically convenient interior approximation for coexistence dynamics.
\end{remark}

Figure~\ref{fig:modeling_bifurcation} contrasts the representative trajectories in the open domain $D=(0,\infty)^2$ under the open-domain formulation, and those in the first quadrant $Q=[0,\infty)^2$ under the absorbed diffusion formulation. 

\begin{figure}[htbp]
    \centering
        \includegraphics[width=\textwidth]{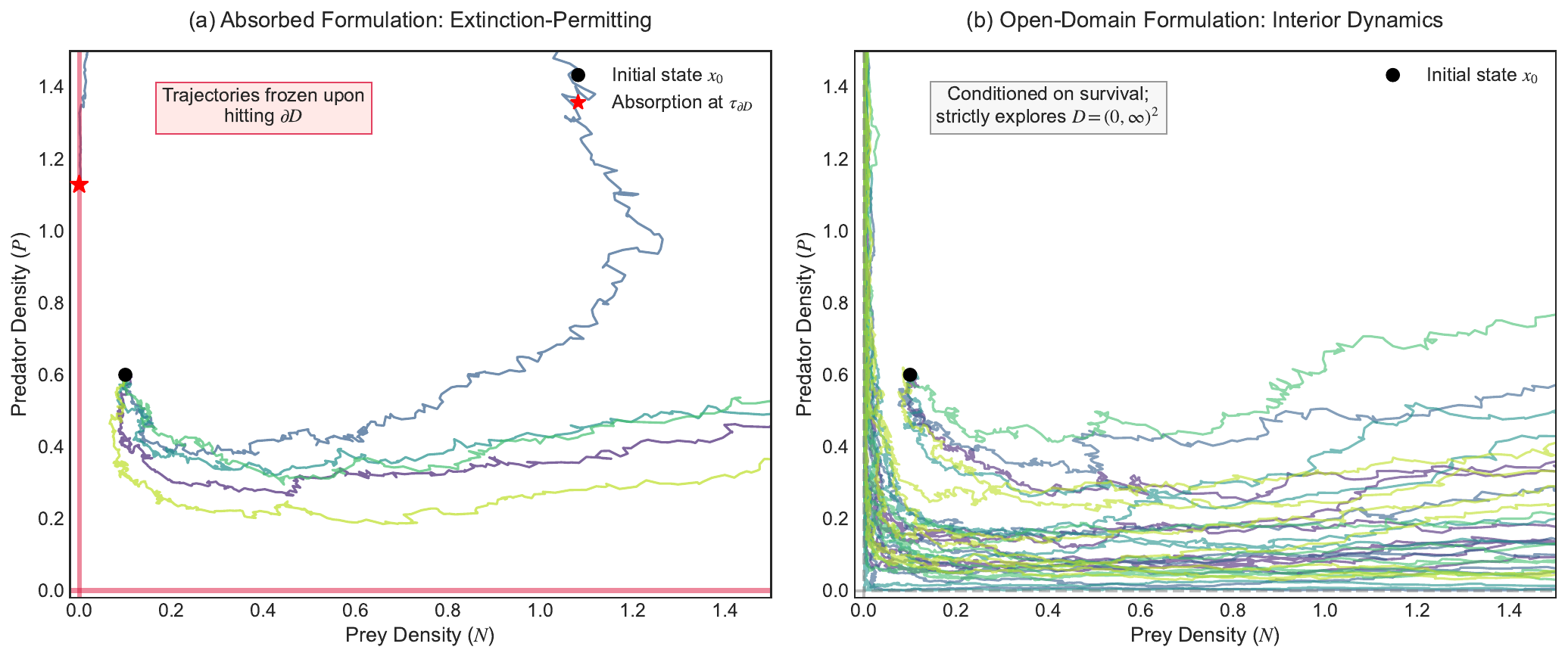}
        \caption{\textbf{Visualizing the modeling bifurcation: Absorbed versus Open-Domain formulations.}
        Both panels simulate the exact same local chemical Langevin drift and covariance parameters ($\Omega=150, k=5.5$, a fluctuation-amplified regime) originating from the identical initial state $x_0$.
        (a) The absorbed formulation, where trajectories are permanently frozen upon first boundary contact ($\tau_{\partial D}$), marked by red stars. This is the mathematically appropriate framework for evaluating extinction probabilities and mean persistence times.
        (b) A numerical proxy for the open-domain formulation, where trajectories are conditioned to survive and continuously explore the interior domain $D=(0,\infty)^2$. This framework captures the geometry of quasi-stationary interior fluctuations but must not be conflated with the extinction-permitting dynamics strictly captured in (a).}
        \label{fig:modeling_bifurcation}
\end{figure}

The next section focuses on the open-domain formulation and develops the corresponding well-posedness architecture on \(D=(0,\infty)^2\): first, maximal strong solvability and non-explosion before boundary hitting; second, positivity invariance under an additional boundary-barrier Lyapunov condition. This two-stage structure is precisely what allows one to be mathematically explicit about the difference between interior-dynamics analysis and extinction-permitting modeling.

\section{Well-posedness on \(D=(0,\infty)^2\): maximal strong solutions, non-explosion, and positivity under barrier conditions}
\label{sec:wellposed_paperA}

This section develops the well-posedness framework for the open-domain diffusion formulation introduced in Section~\ref{sec:formulations_paperA}. The main point is a deliberate separation of statements: we first establish maximal strong solvability together with non-explosion in the interior (up to boundary hitting), and only then give a separate positivity-invariance result under an additional boundary-barrier Lyapunov condition.

We consider the It\^o SDE \eqref{eq:open_domain_SDE_paperA} on the open domain $D$ with the coefficients in \eqref{eq:common_local_coeffs_paperA}.
For the R--M demographic diffusion considered here, the drift \(b\) and covariance \(a\) are smooth on \(D\), while a chosen factorization \(\Sigma\) may fail to be globally Lipschitz near \(\partial D\) because demographic noise intensities vanish as \(N\downarrow 0\) and/or \(P\downarrow 0\). Nevertheless, on every precompact set \(K\Subset D\), the coefficients are locally bounded and (for the factorizations used in this paper) locally Lipschitz. This is the natural setting for localization and Lyapunov arguments.

Define the first exit (boundary hitting) time from the interior by \eqref{eq:tau_boundary_paperA} and freeze $x_t$ according to the rule in \eqref{eq:absorbed_process_definition_paperA}.
The key structural distinction in this section is:
\begin{enumerate}[label=(\roman*)]
\item a statement about maximal strong solvability and no loss of well-posedness in the interior before \(\tau_{\partial D}\), versus
\item a stronger statement about positivity invariance (\(\tau_{\partial D}=\infty\) a.s.), which requires additional assumptions.
\end{enumerate}
This separation is mathematically appropriate for boundary-degenerate demographic diffusions and aligns with the modeling split in Section~\ref{sec:formulations_paperA}: in the absorbed formulation, boundary hitting is allowed. We put the detailed proof for Theorem~\ref{thm:maximal_no_interior_explosion_paperA} and Proposition~\ref{prop:barrier_invariance_paperA} in Appendix~\ref{app:wellposedness_proofs}.

\begin{theorem}[Maximal strong solution and no interior explosion before boundary hitting]
\label{thm:maximal_no_interior_explosion_paperA}
Assume:
\begin{enumerate}[label=(A\arabic*)]
\item \textbf{Local regularity on compacts.} For every precompact set \(K\Subset D\), the coefficients \(b\) and \(\Sigma\) are Lipschitz and bounded on \(K\).
\item \textbf{Lyapunov control at infinity.} There exist \(V_\infty\in C^2(D;[0,\infty))\) and a constant \(C>0\) such that
\begin{equation}
\label{eq:Lyap_infty_cond_paperA}
\mathcal L V_\infty(x)\le C\bigl(1+V_\infty(x)\bigr),\qquad x\in D,
\end{equation}
where
\[
\mathcal L V(x)=\nabla V(x)\cdot b(x)+\frac12\mathrm{tr}\!\bigl(a(x)\nabla^2V(x)\bigr),
\qquad a(x)=\Sigma(x)\Sigma(x)^\top,
\]
and \(V_\infty(x)\to\infty\) whenever \(\|x\|\to\infty\) with \(x\in D\).
\end{enumerate}
Then for every \(x_0\in D\), the SDE~\eqref{eq:open_domain_SDE_paperA} admits a unique maximal strong solution
\[
(x_t)_{0\le t<\tau_*}
\]
with lifetime \(\tau_*\in(0,\infty]\), and
\begin{equation}
\label{eq:no_interior_blowup_paperA}
\mathbb P\bigl(\tau_*<\infty,\ \tau_*<\tau_{\partial D}\bigr)=0.
\end{equation}
Equivalently, loss of well-posedness cannot occur in the interior of \(D\); if \(\tau_*<\infty\), the only possible obstruction is boundary hitting (or boundary approach), i.e.
\begin{equation}
\label{eq:tau_star_equals_boundary_paperA}
\tau_*=\tau_{\partial D}\qquad\text{a.s.}
\end{equation}
\end{theorem}

\begin{proposition}[Positivity invariance under a boundary-barrier Lyapunov condition]
\label{prop:barrier_invariance_paperA}
Assume the hypotheses of Theorem~\ref{thm:maximal_no_interior_explosion_paperA}. In addition, suppose there exists a function
\[
V_b\in C^2(D;[0,\infty))
\]
such that:
\begin{enumerate}[label=\textbf{(B\arabic*)},leftmargin=2.2em]
\item \textbf{Boundary blow-up (barrier property).}
\begin{equation}
\label{eq:barrier_blowup_paperA}
V_b(x)\to\infty\qquad \text{whenever }x\to\partial D\ \text{within }D.
\end{equation}
\item \textbf{Localized generator bound.} There exists \(C_b>0\) such that for every \(n\ge1\) and all \(x\in K_n\),
\begin{equation}
\label{eq:barrier_generator_bound_paperA}
\mathcal L V_b(x)\le C_b\bigl(1+V_b(x)\bigr).
\end{equation}
\end{enumerate}
Then
\begin{equation}
\label{eq:no_boundary_hitting_paperA}
\mathbb P\bigl(\tau_{\partial D}<\infty\bigr)=0.
\end{equation}
Consequently, the maximal strong solution is global and remains in \(D\) for all times:
\begin{equation}
\label{eq:D_invariance_corrected_paperA}
\mathbb P\bigl(x_t\in D\ \text{for all }t\ge0\bigr)=1.
\end{equation}
\end{proposition}

Proposition~\ref{prop:barrier_invariance_paperA} is intentionally separated from Theorem~\ref{thm:maximal_no_interior_explosion_paperA}. The theorem provides the baseline open-domain well-posedness statement: local strong solvability and no interior blow-up before boundary contact. The proposition adds positivity invariance only under an additional barrier condition. This matters conceptually: in the absorbed formulation of Section~\ref{sec:formulations_paperA}, one does not seek to prove \(\tau_{\partial D}=\infty\) a.s.; instead, boundary hitting is allowed, and trajectories are frozen there. Thus, positivity invariance is a property of a particular open-domain analytical regime, not an intrinsic feature of all diffusion interpretations. Figure~\ref{fig:lyapunov_architecture} visualizes the mathematical mechanisms for the prevention of interior blow-up and boundary hitting.

\begin{figure}[htbp]
    \centering
    \includegraphics[width= \textwidth]{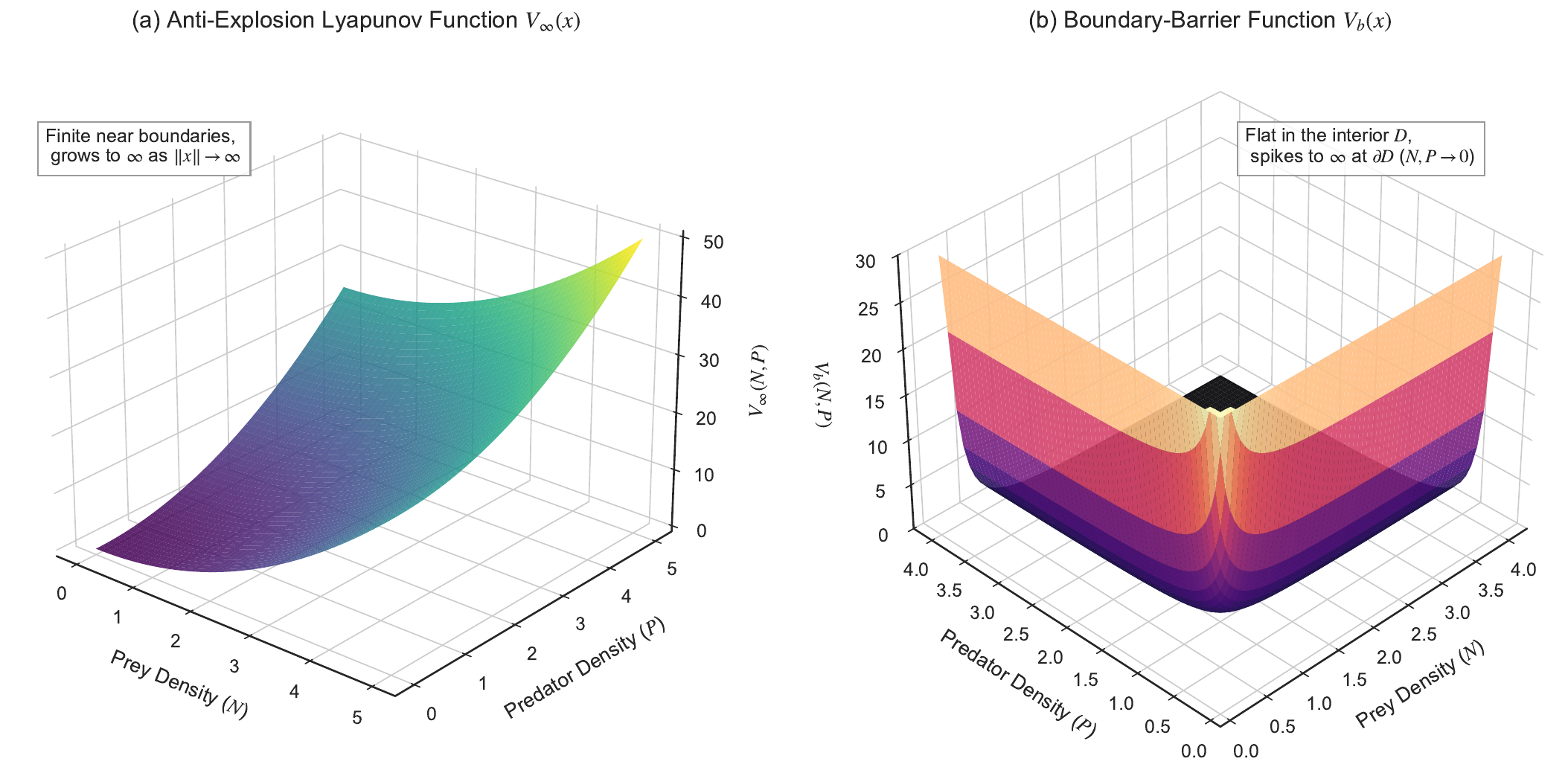}
    \caption{\textbf{Geometric visualization of the two-stage Lyapunov well-posedness architecture.} 
    (a) The anti-explosion Lyapunov function $V_\infty(x)$ (Theorem~\ref{thm:maximal_no_interior_explosion_paperA}) acts as a geometric bowl that grows radially to infinity, strictly confining the stochastic trajectory from escaping to infinite states without restricting its boundary approach. 
    (b) The boundary-barrier Lyapunov function $V_b(x)$ (Proposition~\ref{prop:barrier_invariance_paperA}) acts as a steep geometric cliff. The surface is extremely flat within the interior domain $D=(0,\infty)^2$, permitting natural quasi-stationary fluctuations, but spikes asymptotically to infinity precisely at the boundary faces $N=0$ and $P=0$, theoretically repelling the trajectory and ensuring positivity invariance almost surely.}
    \label{fig:lyapunov_architecture}
\end{figure}

Demographic diffusions often invite comparison with one-dimensional square-root diffusions, where Feller boundary classification \citep{feller_diffusion_1954,peskir_boundary_2015} gives explicit criteria for boundary attainability in terms of drift--diffusion balance. That intuition is useful but incomplete here. In the present two-dimensional predator--prey diffusion:
\begin{enumerate}[label=(\roman*)]
\item drift and diffusion depend on both coordinates via nonlinear interaction terms;
\item the covariance need not be diagonal (indeed, the mechanistic full-covariance closure produces \(a_{12}(x)<0\) in \(D\));
\item the boundary has multiple faces (\(N=0\), \(P=0\)) with distinct biological meanings.
\end{enumerate}
Accordingly, a literal one-dimensional Feller criterion does not transfer directly to the coupled two-dimensional setting. The barrier-based Lyapunov framework above is designed precisely to handle this multivariate, boundary-degenerate structure without forcing an oversimplified 1D analogy.

There is no contradiction between proving positivity invariance under additional assumptions in the open-domain formulation and allowing boundary absorption in the absorbed formulation. These are different mathematical models built from the same local coefficients for different scientific questions. The open-domain formulation supports interior-dynamics analysis (and may, under suitable barriers, remain in \(D\) for all time), while the absorbed formulation supports extinction-permitting analysis by construction. The apparent boundary tension is therefore not a defect but a modeling bifurcation that should be stated explicitly.

\section{Numerical simulation}
\label{sec:numerical_simulation}

This section provides illustrative numerical experiments supporting the structural results established in Sections~\ref{sec:ctmc_cme_paperA}--\ref{sec:wellposed_paperA}. The emphasis is on visualizing the consequences of (i) the mechanistically induced diffusion covariance structure, in particular the predator--prey cross-covariance term, and (ii) the distinction between open-domain and absorbed formulations built from the same local coefficients. 

All simulations use the nondimensional R--M system~\eqref{eq:RM_ODE_paperA} with the Bernoulli-coupled predation--conversion mechanism of Section~\ref{sec:ctmc_cme_paperA}.  Unless otherwise stated, figures are generated from the exact CTMC (Gillespie SSA) or from Euler--Maruyama discretizations of the density-level diffusion under the covariance closures introduced in Section~\ref{sec:diffusion_covariance_paperA}.
Throughout, we fix the baseline parameters
\begin{equation}
\label{eq:sim_baseline_params}
m=1.5,\qquad c=0.4,\qquad e=1,
\end{equation}
and vary the carrying capacity $k$ and system size $\Omega$ as indicated. Under~\eqref{eq:sim_baseline_params}, the Hopf threshold~\eqref{eq:Hopf_threshold_paperA} is
\[
k_H=\frac{m+c}{m-c}=\frac{1.9}{1.1}\approx 1.727.
\]

\subsection{Covariance ellipses: coupled versus split-channel closures}
\label{subsec:sim_ellipses}

Section~\ref{sec:diffusion_covariance_paperA} established that drift-compatible closures can produce different diffusion covariances, and in particular that coupled predation--conversion induces a negative off-diagonal entry $a_{12}(x)<0$, while the split-channel comparator yields $a_{12}^{(S)}(x)=0$ (Proposition~\ref{prop:structural_cross_cov_paperA}). 
We visualize this local geometric difference at the coexistence equilibrium $K_3$ through covariance ellipses.

Figure~\ref{fig:covariance_ellipses} displays $1\sigma$ and $2\sigma$ covariance ellipses at $K_3$ under three closures: (i) the effective coupled $(-1,e)^\top$ closure~\eqref{eq:apred_eff_paperA}, (ii) the exact Bernoulli-coupled closure~\eqref{eq:apred_B_paperA}, and (iii) the drift-compatible split-channel comparator~\eqref{eq:apred_split_paperA}. The two coupled closures produce tilted ellipses, reflecting the negative predator--prey cross-covariance, whereas the split-channel ellipse is axis-aligned. This provides a direct geometric signature of the structural sign result.

The panels also annotate the normalized cross-covariance ratio
\[
\rho(K_3):=\frac{a_{12}(K_3)}{\sqrt{a_{11}(K_3)\,a_{22}(K_3)}}.
\]
As predicted by Section~\ref{sec:diffusion_covariance_paperA}, the Bernoulli-coupled and effective closures share the same off-diagonal term (hence the same sign and similar tilt direction), while differing in the predator variance entry $a_{22}$; see Remark~\ref{rmk:Bernoulli_vs_effective_paperA}.

\begin{figure}[htbp]
\centering
\includegraphics[width=\textwidth]{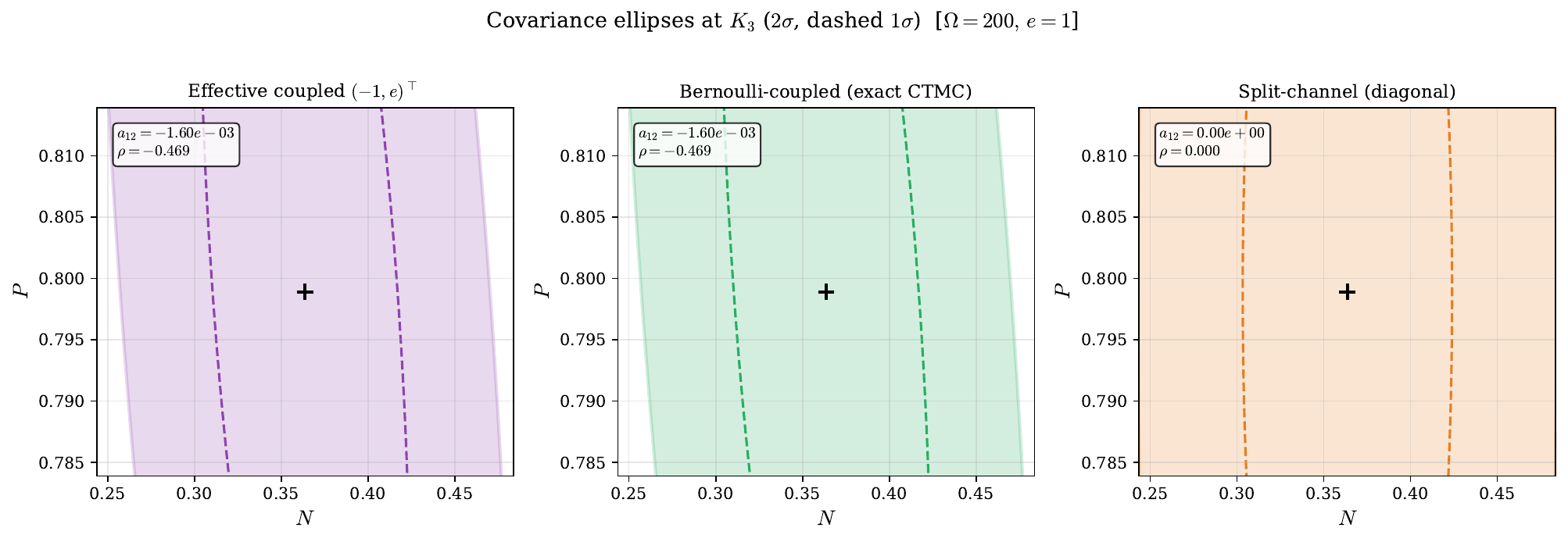}
\caption{Covariance ellipses at $K_3$ under three drift-compatible closures ($\Omega=200$, $e=1$). Left: effective coupled $(-1,e)^\top$. Center: exact Bernoulli-coupled. Right: split-channel (diagonal comparator). Solid and dashed ellipses show $2\sigma$ and $1\sigma$ contours, respectively. The negative cross-covariance tilts the coupled-closure ellipses, while the split-channel ellipse is axis-aligned.}
\label{fig:covariance_ellipses}
\end{figure}

\subsection{Cross-covariance ratio across parameter space}
\label{subsec:sim_heatmap}

To assess when a diagonal-noise approximation may be acceptable in practice (cf.\ Assumption~\ref{ass:diag_condition_paperA} and Remark~\ref{rmk:sufficient_diag_paperA}), we evaluate the normalized cross-covariance ratio $\rho(K_3)$ across the $(m,k)$ plane with $c=0.4$ and $e=1$ fixed.

Figure~\ref{fig:crosscov_heatmap} shows $\rho(K_3)$ under the effective coupled closure. The ratio is uniformly non-positive across the parameter region, in agreement with Proposition~\ref{prop:structural_cross_cov_paperA}, and becomes more negative when $m,k$ increase. The Hopf threshold is overlaid to indicate the partition of $(m,k)$ parameter space into $\Lambda_2$ and $\Lambda_1$.

\begin{figure}[htbp]
\centering
\includegraphics[width=0.7\textwidth]{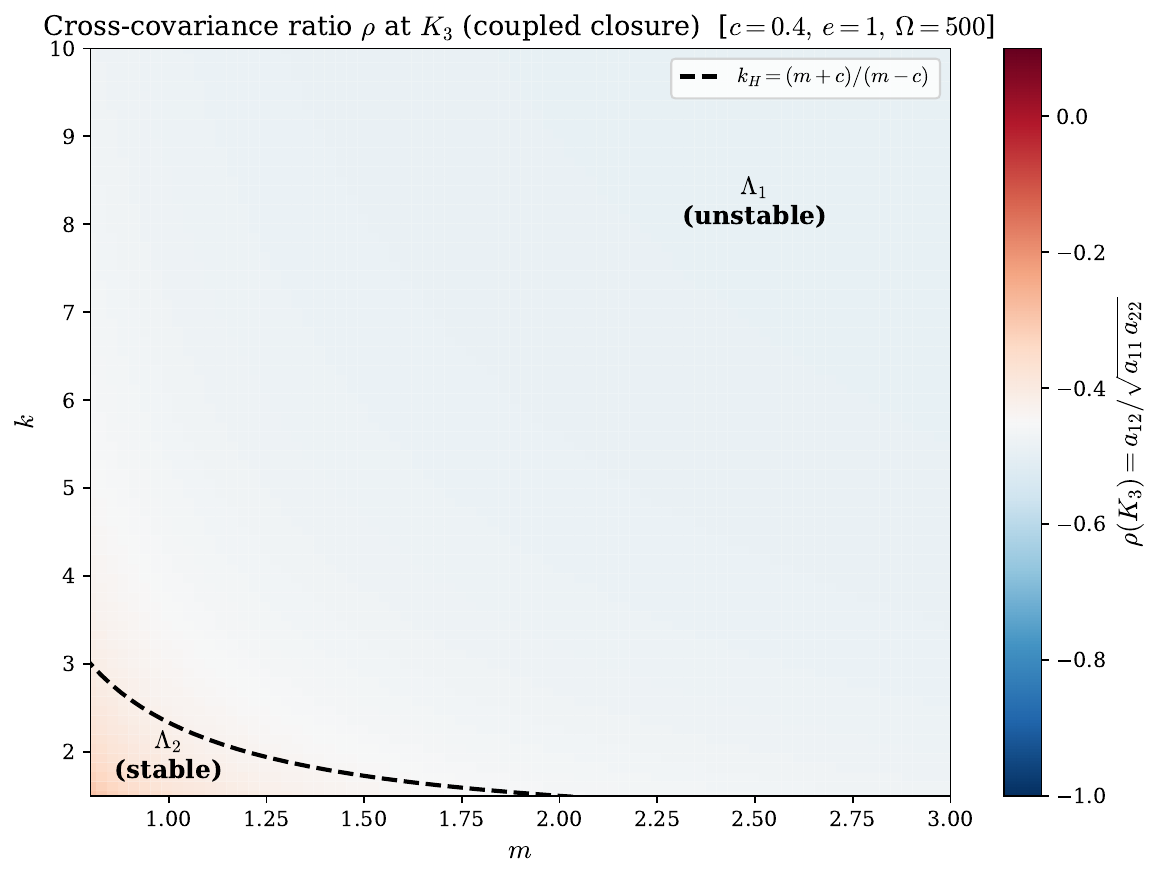}
\caption{Heatmap of $\rho(K_3)=a_{12}/\sqrt{a_{11}a_{22}}$ at coexistence under the effective coupled closure, scanned over $(m,k)$ with $c=0.4$, $e=1$, $\Omega=500$. The dashed curve is the Hopf threshold $k_H=(m+c)/(m-c)$. White regions indicate parameter values for which coexistence is infeasible.}
\label{fig:crosscov_heatmap}
\end{figure}

\subsection{Absorbed versus open-domain boundary behavior}
\label{subsec:sim_boundary}

Section~\ref{sec:formulations_paperA} introduced two distinct global interpretations built from the same local coefficients: an absorbed process for extinction-permitting dynamics and an open-domain formulation for interior dynamics. To visualize the modeling consequences of this distinction, we consider a parameter regime with amplified fluctuations ($\Omega=150$, $k=5.5$) and simulate representative sample paths over a finite horizon.

Figure~\ref{fig:boundary_analysis}(a,d) uses the absorbed formulation, in which trajectories are frozen at first boundary contact, yielding an empirical extinction-time distribution. Figure~\ref{fig:boundary_analysis}(b) shows a positivity-corrected Euler--Maruyama proxy intended to visualize the interior-dynamics perspective associated with the open-domain formulation: when a discrete-time step produces a small sign violation, a numerical clipping near zero is applied to suppress step-size-induced boundary crossing behaviors. This numerical device is used only for visualization and should not be conflated with the continuous-time well-posedness statements proved in Section~\ref{sec:wellposed_paperA}.

The contrast is nevertheless instructive. The absorbed formulation produces a nontrivial distribution of extinction times, while the open-domain proxy emphasizes continued interior fluctuation dynamics. Thus, although the two formulations share the same local drift and covariance in the interior, they answer different modeling questions and produce different observables.

\begin{figure}[htbp]
\centering
\includegraphics[width=0.9\textwidth]{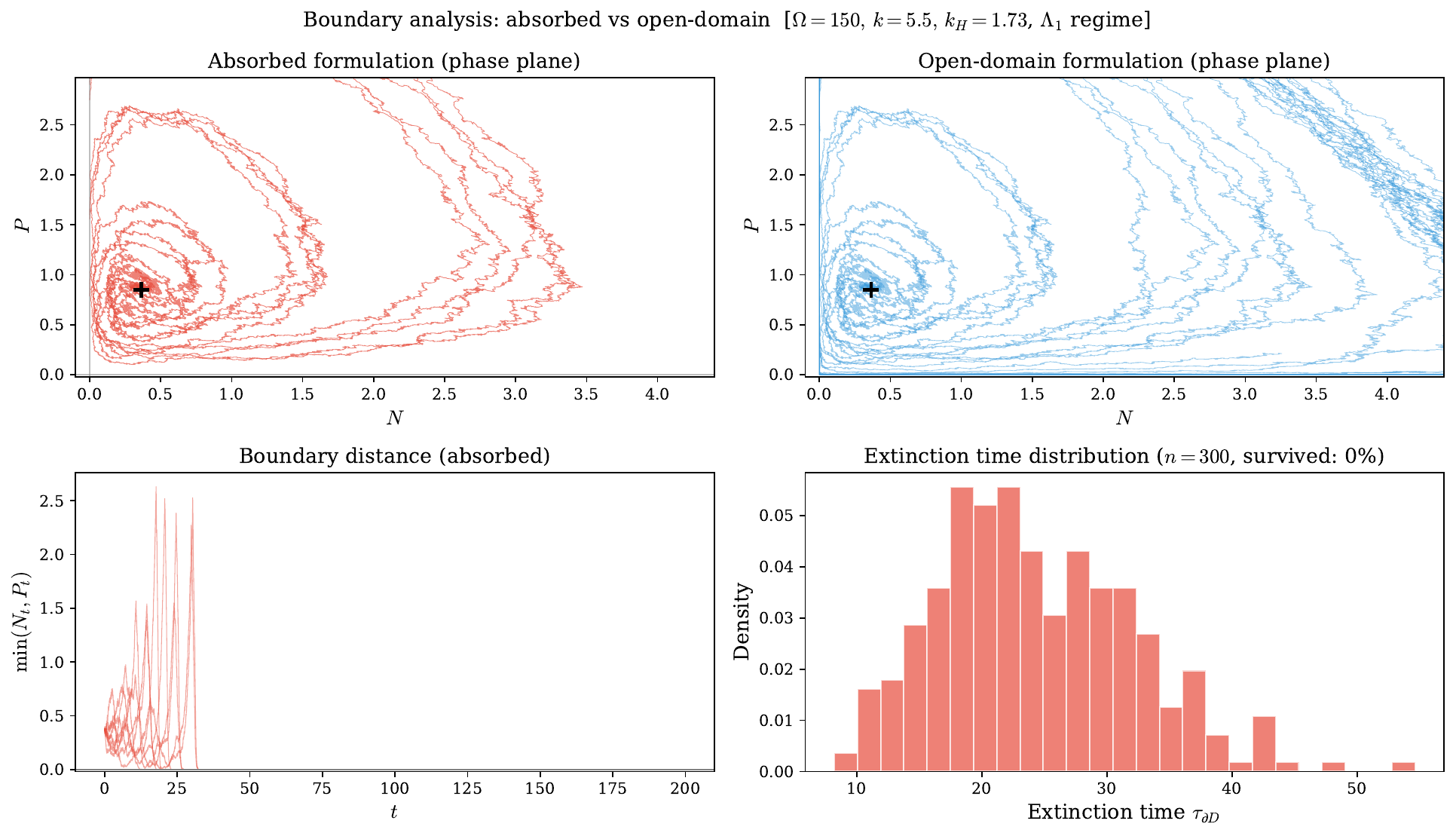}
\caption{Boundary behavior at $\Omega=150$, $k=5.5$ (a fluctuation-amplified regime). (a) Absorbed formulation phase-plane paths. (b) Positivity-corrected Euler--Maruyama proxy for the open-domain interior-dynamics perspective. (c) Boundary-distance diagnostic $\min(N_t,P_t)$ for absorbed paths. (d) Empirical extinction-time distribution from $n=300$ absorbed simulations.}
\label{fig:boundary_analysis}
\end{figure}

\subsection{Parameter-regime visualization relative to the Hopf threshold}
\label{subsec:sim_hopf}

Finally, we provide a brief visual comparison of stochastic trajectories across the parameter partition induced by the Hopf threshold~\eqref{eq:Hopf_threshold_paperA}. Figure~\ref{fig:hopf_regimes} compares full-covariance diffusion simulations (with deterministic trajectories overlaid) for $k=0.7\,k_H$, $k=k_H$, and $k=1.5\,k_H$.

The purpose of this figure is primarily organizational: it illustrates that the partition into $\Lambda_2$ and $\Lambda_1$ corresponds to visibly different fluctuation regimes (stable-coexistence fluctuations, threshold-amplified oscillatory fluctuations, and noise-perturbed cycling, respectively).

\begin{figure}[htbp]
\centering
\includegraphics[width=0.9\textwidth]{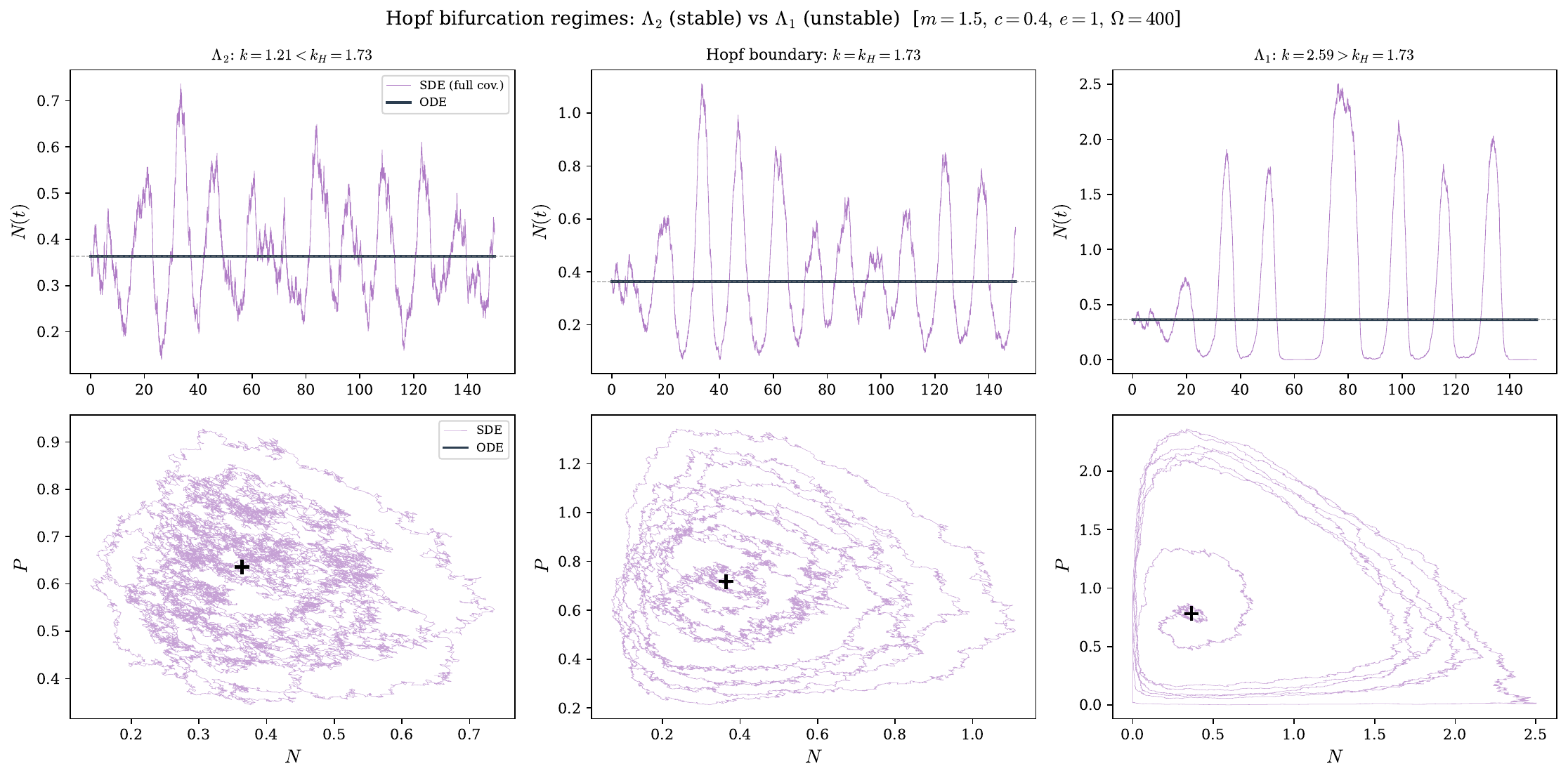}
\caption{Trajectory regimes relative to the Hopf threshold at $\Omega=400$. Left: $k=0.7\,k_H$ (inside $\Lambda_2$). Center: $k=k_H$ (threshold). Right: $k=1.5\,k_H$ (inside $\Lambda_1$). Top row: prey time series (diffusion vs.\ ODE). Bottom row: phase-plane trajectories.}
\label{fig:hopf_regimes}
\end{figure}

Taken together, these numerical experiments reinforce the structural conclusions developed in the analytical sections. The covariance-ellipse visualizations provide a direct geometric manifestation of the mechanistically induced negative cross-covariance; the parameter-space scan shows that this structural sign persists across the parameter sweeps; and the boundary comparisons clarify that absorbed and open-domain formulations, though locally identical, encode distinct modeling objectives. The simulations, therefore, serve not as independent evidence but as visual diagnostics of the theoretical architecture: they illustrate how event-channel design, covariance structure, and boundary interpretation propagate from microscopic assumptions to macroscopic stochastic behavior.

\section{Discussion}
\label{sec:discussion_paperA}

This paper develops a mechanistically consistent stochastic framework for the R--M predator--prey model.
It demonstrates that diffusion-level drift and covariance must be rigorously derived from an integer-valued CTMC/CME description rather than imposed directly at the SDE level. 
Ultimately, the value of this work lies not in proposing yet another stochastic variant of the R--M model, but in establishing a mathematically rigorous bridge linking:
(i) microscopic event stoichiometry to macroscopic demographic covariance,
(ii) open-domain survival to absorbed-boundary extinction, and 
(iii) boundary-degenerate local coefficients to a globally tailored well-posedness architecture.

\paragraph{Mechanistic covariance structure as a primary result (Main result 1).}
Our first principal contribution isolates the macroscopic consequences of microscopic event coupling.
Starting from the Bernoulli-coupled predation--conversion CTMC, the diffusion covariance is strictly inherited through the CME-consistent identity
\[
a(x)=\frac{1}{\Omega}\,S\,\mathrm{diag}(f(x))\,S^\top.
\]
Within this framework, the predator--prey cross-covariance is not a tunable modeling embellishment: it is an unavoidable structural output of the stoichiometric coupling between prey loss and predator gain during predation--conversion events. 
In particular, the coupled closures systematically studied in Section~\ref{sec:diffusion_covariance_paperA} strictly produce a negative predation-induced off-diagonal term.
This rigorously exposes the fallacy that drift equivalence implies covariance equivalence, 
demonstrating that diagonal-noise closures must be relegated to regime-specific approximations rather than treated as default mechanistic baselines.

\paragraph{Boundary-aware formulation distinction as a modeling result (Main result 2).}
The second principal result is the formal bifurcation of two global diffusion formulations built from the exact same local coefficients.
One is an open-domain formulation on \(D=(0,\infty)^2\), used strictly for interior, survival-conditioned dynamics.
The other is an absorbed formulation, which allows for extinction-permitting dynamics. 
While routinely blurred or entirely ignored in the stochastic ecological literature, this distinction is both mathematically profound and biologically consequential. 
The open-domain formulation provides the precise analytical arena for coexistence fluctuations and local spectral geometry.
Conversely, the absorbed formulation is the exclusive framework for evaluating extinction probabilities and mean persistence times, mirroring the irreversibility of CTMC extinction.
By stating this bifurcation explicitly, our framework eradicates a pervasive source of modeling ambiguity.

\paragraph{Two-stage well-posedness architecture as an analytical result (Main result 3).}
The third principal result is the customized two-stage well-posedness architecture for the open-domain SDE. Section~\ref{sec:wellposed_paperA} deliberately separates:
\begin{enumerate}[label=(\roman*)]
\item maximal strong solvability together with non-explosion in the interior before boundary hitting, and
\item positivity invariance under an additional boundary-barrier Lyapunov condition.
\end{enumerate}
This separation is not a mere technical bookkeeping device; it is a vital architectural safeguard.
It explicitly aligns the mathematical theorems with the biological modeling distinction formalized in Section~\ref{sec:formulations_paperA}.
Positivity invariance is a property that may hold in a specific open-domain analytical regime, 
but it is not a default mathematical right one can claim when the intended model is extinction-permitting. 
This formulation thereby avoids analytical overreach and provides a pristine mathematical template for other boundary-degenerate population diffusions.

\paragraph{Scope and limitations.}
We contextualize these contributions within several recognized limitations.
First, the diffusion model remains a second-order approximation to the exact CTMC/CME. 
It is therefore most robust in regimes where the system-size parameter \(\Omega\) is sufficiently large and event counts over observational windows are strictly non-sparse. 
For severely small populations, integer-valued effects and immediate boundary absorption dominate, 
rendering exact CTMC/SSA descriptions as the mandatory baseline.
Second, the present paper emphasizes exact structural derivation and well-posedness architecture rather than quantitative diffusion error bounds; 
no uniform-in-time approximation bound in \(\Omega\) is pursued here. 
Third, the open-domain analysis is formulated at the abstract level of local coefficients and Lyapunov or barrier criteria.
Verifying these conditions in highly specific or extreme parameter regimes may require bespoke analytical constructions beyond the generic bounding techniques presented.

\paragraph{Generalizations and portability of the framework.}
Crucially, the analytical framework developed here is highly portable and extends well beyond the R--M model.
The core covariance identity
\[
a(x)=\frac{1}{\Omega}\,S\,\mathrm{diag}(f(x))\,S^\top
\]
depends universally on event stoichiometry and intensity functions, 
completely agnostic to the specific Holling type~II form.
Thus, this exact blueprint can be deployed for alternative predator--prey mechanisms (e.g., ratio-dependent interactions) and seamlessly scaled to high-dimensional multi-species food-web models.
More broadly, the integration of biological mechanisms, rigorous mathematical modeling, and threshold-driven organization reflects a portable analytical paradigm that extends to diverse biological systems beyond ecological interactions \citep{liang_global_2025,wang_damage-structured_2026}.
In every such extension, the same three modeling questions reappear.
First, which event couplings mandate nontrivial cross-covariances? 
Second, which specific boundary faces represent biologically distinct extinction cascades?
Third, which open-domain versus absorbed interpretations align with the target observables?
This paper provides the foundational, reusable template for answering those questions explicitly.

\paragraph{Conclusions.}
The central thesis of this paper is that demographic covariance in predator--prey diffusion models must be treated as a mechanistic output, not a stylistic choice.
For the R--M system, predation--conversion coupling structurally mandates a negative predator--prey cross-covariance.
This feature interacts fundamentally with boundary interpretation and stochastic well-posedness. 
We derive the diffusion directly from an integer-valued CTMC model, explicitly separate open-domain from absorbed formulations, and formalize a bespoke two-stage well-posedness architecture.
Together, these steps establish a mathematically rigorous foundation for covariance-consistent, boundary-aware stochastic ecological modeling.

\begin{appendices}

\section{From the CME to the diffusion approximation}
\label{app:CME_diffusion}

This appendix records a standard but explicit route from a mechanistic CTMC, via the CME, to the diffusion approximation used in the main text. It also clarifies the origin of the covariance identity
\begin{equation}\label{eq:appA_key}
a(x)=\Sigma(x)\Sigma(x)^\top
=\frac{1}{\Omega}\,S\,\mathrm{diag}\big(f(x)\big)\,S^\top,
\end{equation}
which underlies the full-covariance chemical Langevin formulation. At the count level, the corresponding infinitesimal covariance is
\begin{equation}\label{eq:appA_key_counts}
G(X)=S\,\mathrm{diag}\big(\lambda(X)\big)\,S^\top.
\end{equation}

For the predator--prey application in Section~\ref{sec:ctmc_cme_paperA}, the key message is that the diffusion matrix is not an ad hoc modeling choice: it is inherited directly from the stoichiometry and propensities of the underlying CTMC/CME description.

\subsection{CTMC formulation and the CME}
\label{app:CME_setup}

Let \(X(t)\in\mathbb N_0^d\) denote the vector of population counts (in the main text, \(d=2\) for prey and predator). Consider \(K\) reaction/event channels indexed by \(k=1,\dots,K\), with stoichiometric increments \(\nu_k\in\mathbb Z^d\). Collect them in the stoichiometric matrix
\[
S=[\nu_1\ \nu_2\ \cdots\ \nu_K]\in\mathbb Z^{d\times K}.
\]
When channel \(k\) fires at state \(X\), the process jumps according to
\[
X \longmapsto X+\nu_k.
\]

Let \(\lambda_k(X)\ge 0\) be the propensity of channel \(k\) at state \(X\). By definition, for a short interval \([t,t+dt)\),
\begin{equation}\label{eq:appA_propensity}
\mathbb P(\text{channel \(k\) fires in }[t,t+dt)\mid X(t)=X)
=
\lambda_k(X)\,dt+o(dt),
\end{equation}
and, up to \(o(dt)\), at most one event occurs.

Let \(p(X,t):=\mathbb P(X(t)=X)\). Conditioning on what occurs during \([t,t+dt)\), there are (up to \(o(dt)\)) two contributions to \(p(X,t+dt)\):
\begin{enumerate}[label=(\roman*)]
\item no event occurs and the system was already in state \(X\), contributing
\begin{equation}\label{eq:appA_CME_contrib1}
p(X,t)\Big(1-\sum_{k=1}^K\lambda_k(X)\,dt\Big)+o(dt);
\end{equation}
\item exactly one event \(k\) occurs and moves the system from \(X-\nu_k\) into \(X\), contributing
\begin{equation}\label{eq:appA_CME_contrib2}
\sum_{k=1}^K p(X-\nu_k,t)\,\lambda_k(X-\nu_k)\,dt+o(dt).
\end{equation}
\end{enumerate}
Summing these terms, subtracting \(p(X,t)\), dividing by \(dt\), and letting \(dt\downarrow 0\) yields the CME:
\begin{equation}\label{eq:appA_CME_general}
\frac{\partial}{\partial t}p(X,t)
=
\sum_{k=1}^K
\Big[
\lambda_k(X-\nu_k)\,p(X-\nu_k,t)
-
\lambda_k(X)\,p(X,t)
\Big].
\end{equation}

\subsection{Generator and moment identities}
\label{app:CME_moments}

Let \(\varphi:\mathbb N_0^d\to\mathbb R\) be a test function of suitable growth. Multiplying~\eqref{eq:appA_CME_general} by \(\varphi(X)\), summing over states, and using index shifts gives
\begin{equation}\label{eq:appA_generator_identity}
\frac{d}{dt}\,\mathbb E[\varphi(X(t))]
=
\mathbb E\!\left[(\mathcal A\varphi)(X(t))\right],
\end{equation}
where the CTMC generator is
\begin{equation}\label{eq:appA_CTMC_generator}
(\mathcal A\varphi)(X)
=
\sum_{k=1}^K \lambda_k(X)\,\big(\varphi(X+\nu_k)-\varphi(X)\big).
\end{equation}

We compute the first moment (drift).
Taking \(\varphi(X)=X_i\) gives
\[
\varphi(X+\nu_k)-\varphi(X)=(\nu_k)_i,
\]
hence
\begin{equation}\label{eq:appA_first_moment}
\frac{d}{dt}\,\mathbb E[X_i(t)]
=
\mathbb E\!\left[\sum_{k=1}^K \lambda_k(X(t))\,(\nu_k)_i\right].
\end{equation}
In vector form, this is the exact CTMC drift identity
\begin{equation}\label{eq:appA_first_moment_vector}
\frac{d}{dt}\,\mathbb E[X(t)]
=
\mathbb E\!\left[S\,\lambda(X(t))\right],
\end{equation}
where \(\lambda(X)=(\lambda_1(X),\dots,\lambda_K(X))^\top\).

We compute the second moment and quadratic variation structure.
Taking \(\varphi(X)=X_iX_j\) and using
\[
(X+\nu_k)_i(X+\nu_k)_j-X_iX_j
=
X_i(\nu_k)_j+X_j(\nu_k)_i+(\nu_k)_i(\nu_k)_j,
\]
we obtain
\begin{equation}\label{eq:appA_second_moment_raw}
\frac{d}{dt}\,\mathbb E[X_i(t)X_j(t)]
=
\mathbb E\!\left[
\sum_{k=1}^K \lambda_k(X(t))
\Big(
(X(t))_i(\nu_k)_j+(X(t))_j(\nu_k)_i+(\nu_k)_i(\nu_k)_j
\Big)
\right].
\end{equation}
The term \((\nu_k)_i(\nu_k)_j\) is the key source of infinitesimal covariance; in matrix form it appears as \(\nu_k\nu_k^\top\).

\subsection{Infinitesimal covariance at the CTMC level}
\label{app:CME_inf_cov}

Let \(\Delta X := X(t+dt)-X(t)\). Conditioning on \(X(t)=X\), equation~\eqref{eq:appA_propensity} implies
\[
\mathbb E[\Delta X\mid X(t)=X]
=
\sum_{k=1}^K \lambda_k(X)\,\nu_k\,dt + o(dt),
\]
and
\[
\mathbb E[\Delta X\Delta X^\top\mid X(t)=X]
=
\sum_{k=1}^K \lambda_k(X)\,\nu_k\nu_k^\top\,dt + o(dt).
\]
Therefore,
\begin{align}
\mathrm{Cov}(\Delta X\mid X(t)=X)
&=
\mathbb E[\Delta X\Delta X^\top\mid X]
-
\mathbb E[\Delta X\mid X]\mathbb E[\Delta X\mid X]^\top \notag\\
&=
\sum_{k=1}^K \lambda_k(X)\,\nu_k\nu_k^\top\,dt + o(dt), \label{eq:appA_cov_short_time}
\end{align}
since the outer product of the conditional mean is \(O(dt^2)\).

Hence, the exact infinitesimal covariance per unit time at the count scale is
\begin{equation}\label{eq:appA_inf_cov_counts}
G(X)
:=
\lim_{dt\downarrow 0}\frac{1}{dt}\,\mathrm{Cov}(\Delta X\mid X(t)=X)
=
\sum_{k=1}^K \lambda_k(X)\,\nu_k\nu_k^\top
=
S\,\mathrm{diag}(\lambda(X))\,S^\top.
\end{equation}
This identity is exact for the CTMC and is the source of the diffusion covariance used in the chemical Langevin approximation.

\subsection{Density-dependent scaling and the diffusion matrix}
\label{app:KM_FP}

Introduce a system-size parameter \(\Omega\gg 1\) and define density variables
\[
x:=\frac{X}{\Omega}\in\mathbb R_+^d.
\]
Assume density-dependent propensities of Kurtz type:
\begin{equation}\label{eq:appA_density_dependent_prop}
\lambda_k(X)
=
\Omega\,f_k\!\left(\frac{X}{\Omega}\right)
=
\Omega\,f_k(x),
\qquad k=1,\dots,K,
\end{equation}
for smooth rate functions \(f_k:\mathbb R_+^d\to\mathbb R_+\). Let \(f(x)=(f_1(x),\dots,f_K(x))^\top\).

Under the standard chemical Langevin scaling, the density process is approximated by a diffusion with drift and covariance
\begin{equation}\label{eq:appA_drift_diff_density}
b(x)=S\,f(x),
\qquad
a(x)=\frac{1}{\Omega}\,S\,\mathrm{diag}\big(f(x)\big)\,S^\top.
\end{equation}
The same formula follows directly by rescaling the CTMC short-time covariance~\eqref{eq:appA_cov_short_time}: since \(\Delta x=\Delta X/\Omega\) and $x = X/\Omega$, 
\begin{align}
\mathrm{Cov}(\Delta x\mid x)
&=
\frac{1}{\Omega^2}\,\mathrm{Cov}(\Delta X\mid X)\notag\\
&=
\frac{1}{\Omega^2}
\left(
\sum_{k=1}^K \lambda_k(X)\,\nu_k\nu_k^\top\,dt
\right)
+o\!\left(\frac{dt}{\Omega}\right)\notag\\
&=
\frac{1}{\Omega}
\left(
\sum_{k=1}^K f_k(x)\,\nu_k\nu_k^\top
\right)dt
+o\!\left(\frac{dt}{\Omega}\right)\notag\\
&=
a(x)\,dt+o\!\left(\frac{dt}{\Omega}\right). \label{eq:appA_cov_density_scale}
\end{align}
Thus, the diffusion matrix \(a(x)\) is precisely the density-scale counterpart of the exact count-scale infinitesimal covariance \(G(X)\).

\section{Proofs for the open-domain well-posedness results in Section~\ref{sec:wellposed_paperA}}
\label{app:wellposedness_proofs}

This appendix provides detailed proofs of the main well-posedness statements in Section~\ref{sec:wellposed_paperA}. We work on the open state space
\[
D=(0,\infty)^2\subset\mathbb R^2,
\]
and consider the It\^o SDE
\begin{equation}
\label{eq:appC_sde}
dx_t=b(x_t)\,dt+\Sigma(x_t)\,dW_t,\qquad x_0\in D,
\end{equation}
where \(W\) is an \(r\)-dimensional Brownian motion and \(a(x)=\Sigma(x)\Sigma(x)^\top\).

For convenience, we restate the generator acting on \(V\in C^2(D)\):
\begin{equation}
\label{eq:appC_generator}
\mathcal LV(x)=\nabla V(x)\cdot b(x)+\frac12\mathrm{tr}\!\big(a(x)\nabla^2V(x)\big).
\end{equation}
We also recall the first boundary-hitting time
\begin{equation}
\label{eq:appC_tau_boundary}
\tau_{\partial D}:=\inf\{t\ge 0:\ x_t\notin D\}.
\end{equation}

\subsection{Preliminaries: exhaustion, localization, and a geometric lemma}
\label{appC:preliminaries}

We begin with standard localization objects and a geometric fact used in the proof of Theorem~\ref{thm:maximal_no_interior_explosion_paperA}.

We detail exhaustion of \(D\) by precompact sets.
Fix an increasing exhaustion \((K_n)_{n\ge 1}\) of \(D\) by compact sets with nonempty interior such that
\begin{equation}
\label{eq:appC_exhaustion}
K_n\subset \operatorname{int}(K_{n+1}),\qquad K_n\Subset D,\qquad \bigcup_{n=1}^\infty K_n=D.
\end{equation}
A concrete choice is
\[
K_n:=\Big[\frac1n,n\Big]^2\cap \overline{B(0,n)}.
\]
For a continuous adapted process \(x\), define the exit times
\begin{equation}
\label{eq:appC_tau_n}
\tau_n:=\inf\{t\ge 0:\ x_t\notin K_n\}.
\end{equation}
Then \(\tau_n\) is increasing in \(n\), and we set
\begin{equation}
\label{eq:appC_tau_star_def}
\tau_*:=\lim_{n\to\infty}\tau_n\in(0,\infty].
\end{equation}

\begin{lemma}[Leaving every compact subset of \(D\) while norm-bounded implies boundary approach]
\label{lem:appC_boundary_approach}
Let \(D=(0,\infty)^2\). Suppose \((x_n)_{n\ge1}\subset D\) satisfies
\[
\sup_{n\ge1}\|x_n\|<\infty.
\]
Furthermore, for every compact $K\Subset D$, only finitely many $x_n$ lie in $K$. 
Then
\[
\operatorname{dist}(x_n,\partial D)\to 0
\quad\text{along a subsequence}.
\]
Equivalently, there exists a subsequence \((x_{n_j})\) with \(x_{n_j}\to x_\infty\in\partial D\) after passing to a further subsequence.
\end{lemma}

\begin{proof}
Since \((x_n)\) is norm-bounded in \(\mathbb R^2\), Bolzano--Weierstrass theorem yields a convergent subsequence \(x_{n_j}\to x_\infty\in\overline D=[0,\infty)^2\). 
Suppose, for contradiction, that $x_{\infty}\in D$. Then there exists $\varepsilon>0$ such that 
\(\overline{B(x_\infty,\varepsilon)}\subset D\). 
By convergence, for sufficiently large $j$ we have \(x_{n_j}\in \overline{B(x_\infty,\varepsilon)}\), so the infinite tail $\{x_{n_j}:j\ge J\}$ is contained in the compact set \(\overline{B(x_\infty,\varepsilon)}\Subset D\). 
This contradicts the standing assumption that for every compact \(K\Subset D\) only finitely many terms of $(x_n)$ belong to $K$. Hence \(x_\infty\in\partial D\), which implies \(\operatorname{dist}(x_{n_j},\partial D)\to 0\).
\end{proof}

\subsection{Complete proof of Theorem~\ref{thm:maximal_no_interior_explosion_paperA}}
\label{appC:proof_theorem_maximal}

\begin{proof}
The proof is divided into five steps.

\medskip
\noindent\textbf{Step 1: Local strong solutions on each localization domain.}
Fix \(n\ge1\). By assumption \textbf{(A1)}, the coefficients \(b,\Sigma\) are Lipschitz and bounded on \(K_n\Subset D\). Extend \(b,\Sigma\) to globally Lipschitz, bounded coefficients \(b^{(n)},\Sigma^{(n)}\) on \(\mathbb R^2\) such that
\[
b^{(n)}(x)=b(x),\qquad \Sigma^{(n)}(x)=\Sigma(x)\qquad \text{for }x\in K_n.
\]
(For example, use a cutoff and McShane-type extension componentwise.)
Then the globally defined SDE
\begin{equation}
\label{eq:appC_sde_extended}
dX_t^{(n)}=b^{(n)}(X_t^{(n)})\,dt+\Sigma^{(n)}(X_t^{(n)})\,dW_t,\qquad X_0^{(n)}=x_0,
\end{equation}
admits a unique global strong solution by standard It\^o SDE theory.

Define the exit time from \(K_n\) for this solution:
\[
\tau_n^{(n)}:=\inf\{t\ge0:\ X_t^{(n)}\notin K_n\}.
\]
On \([0,\tau_n^{(n)})\), \(X_t^{(n)}\in K_n\), so \(b^{(n)}=b\) and \(\Sigma^{(n)}=\Sigma\). Therefore \(X^{(n)}\) solves the original SDE~\eqref{eq:appC_sde} up to \(\tau_n^{(n)}\).

\medskip
\noindent\textbf{Step 2: Consistency and pathwise uniqueness patching.}
Let \(m>n\). Consider the extended solutions \(X^{(n)}\) and \(X^{(m)}\) driven by the same Brownian motion and the same initial condition. On \(K_n\), both sets of extended coefficients agree with the original coefficients \((b,\Sigma)\), and \textbf{(A1)} yields local Lipschitz continuity on \(K_n\). By pathwise uniqueness for locally Lipschitz SDEs up to the exit time from \(K_n\), we have
\[
X_t^{(n)}=X_t^{(m)}\qquad \text{for all }t<\tau_n^{(n)}\wedge \tau_n^{(m)}\quad\text{a.s.}
\]
Hence the stopped processes are consistent across \(n\), and we may define a process \(x_t\) on \(t<\tau_*\) by patching:
\[
x_t:=X_t^{(n)}\qquad\text{whenever }t<\tau_n,
\]
where \(\tau_n\) is the (common) exit time from \(K_n\) under the patched process. This yields a strong solution of~\eqref{eq:appC_sde} on \([0,\tau_*)\), with
\[
\tau_*=\lim_{n\to\infty}\tau_n.
\]
By the same local uniqueness argument, this maximal solution is pathwise unique up to \(\tau_*\).

\medskip
\noindent\textbf{Step 3: A Lyapunov estimate up to bounded-radius localization.}
Fix \(R>0\) and define
\begin{equation}
\label{eq:appC_sigma_R}
\sigma_R:=\inf\{t\ge0:\ \|x_t\|\ge R\}\wedge \tau_{\partial D}.
\end{equation}
For each \(R\), the stopped process \(x_{t\wedge \sigma_R}\) remains in the bounded set \(D\cap \overline{B(0,R)}\), so It\^o's formula applies to \(V_\infty(x_{t\wedge \sigma_R})\). Using \eqref{eq:appC_generator} and assumption \textbf{(A2)},
\begin{align}
V_\infty(x_{t\wedge \sigma_R})
&=V_\infty(x_0)+\int_0^{t\wedge \sigma_R}\mathcal LV_\infty(x_s)\,ds + M_{t\wedge \sigma_R}\notag\\
&\le V_\infty(x_0)+ C\int_0^{t\wedge \sigma_R}(1+V_\infty(x_s))\,ds + M_{t\wedge \sigma_R},
\label{eq:appC_ito_Vinf}
\end{align}
where \(M_{t\wedge \sigma_R}\) is a martingale with zero expectation. Taking expectations gives
\begin{equation}
\label{eq:appC_lyap_expectation_pre_gronwall}
\mathbb E[V_\infty(x_{t\wedge \sigma_R})]
\le
V_\infty(x_0)+C t + C\int_0^t \mathbb E[V_\infty(x_{s\wedge \sigma_R})]\,ds.
\end{equation}
By Gr\"onwall's inequality,
\begin{equation}
\label{eq:appC_lyap_expectation_gronwall}
\mathbb E[V_\infty(x_{t\wedge \sigma_R})]
\le e^{Ct}\bigl(V_\infty(x_0)+Ct\bigr),\qquad t\ge0.
\end{equation}
The right-hand side is independent of \(R\).

\medskip
\noindent\textbf{Step 4: No explosion to infinity before boundary hitting.}
We show that the process cannot lose well-posedness in the interior by escaping to infinity before hitting \(\partial D\). Suppose, for contradiction, that there exists \(T>0\) such that
\[
\mathbb P\bigl(\tau_* \le T,\ \tau_*<\tau_{\partial D},\ \sup_{t<\tau_*}\|x_t\|=\infty \bigr)>0.
\]
On the event inside the probability, \(\sigma_R\le \tau_*\wedge T\) for all sufficiently large \(R\), and \(\|x_{\sigma_R}\|=R\). Since \(\tau_*<\tau_{\partial D}\), the process remains in \(D\) up to \(\tau_*\), so \(\|x_{\sigma_R}\|\to\infty\) within \(D\). By assumption \textbf{(A2)},
\[
V_\infty(x_{\sigma_R})\to\infty
\quad\text{on that event.}
\]
Thus \(V_\infty(x_{T\wedge \sigma_R})\to\infty\) on a set of positive probability. By Fatou's lemma, \(\sup_R \mathbb E[V_\infty(x_{T\wedge \sigma_R})]=\infty\), contradicting the uniform estimate \eqref{eq:appC_lyap_expectation_gronwall}. Therefore,
\begin{equation}
\label{eq:appC_no_infinity_escape}
\mathbb P\bigl(\tau_*<\infty,\ \tau_*<\tau_{\partial D},\ \sup_{t<\tau_*}\|x_t\|=\infty\bigr)=0.
\end{equation}

\medskip
\noindent\textbf{Step 5: If every compact set is eventually left while the norm remains bounded, the path approaches $\partial D$.}

It remains to rule out the possibility that
\[
\tau_*<\infty,\qquad \tau_*<\tau_{\partial D},\qquad 
\sup_{t<\tau_*}\|x_t\|<\infty .
\]
On this event, since $\tau_*=\lim_{n\to\infty}\tau_n$ with
$\tau_n=\inf\{t:\,x_t\notin K_n\}$ and $K_n\Subset D$ increasing to $D$,
the trajectory leaves each compact $K_n$ before time $\tau_*$. Hence we may
choose times $t_n\uparrow\tau_*$ such that
\[
x_{t_n}\notin K_n .
\]

Because $\sup_{t<\tau_*}\|x_t\|<\infty$, the sequence $(x_{t_n})\subset D$
is norm-bounded. Moreover, for any compact $K\Subset D$ there exists $N$
such that $K\subset K_N$; therefore $x_{t_n}\notin K$ for all $n\ge N$.
Thus $(x_{t_n})$ is eventually outside every compact subset of $D$.

By Lemma~\ref{lem:appC_boundary_approach}, any norm-bounded sequence in $D$
that eventually leaves every compact subset of $D$ admits a subsequence whose
distance to $\partial D$ tends to zero. Consequently, some subsequence of
$(x_{t_n})$ approaches $\partial D$.

Since $x_t$ is continuous on $[0,\tau_*)$, it follows that
\[
\inf_{0\le t<\tau_*}\operatorname{dist}(x_t,\partial D)=0,
\]
that is, the path approaches the boundary as $t\uparrow\tau_*$.

This contradicts the event $\tau_*<\tau_{\partial D}$, which asserts that no
boundary hit or exit occurs before time $\tau_*$. Hence any loss of compactness
in $D$ while the norm remains bounded must correspond to approach to the
boundary rather than interior breakdown.

Combining this with \eqref{eq:appC_no_infinity_escape}, we conclude
\[
\mathbb P\bigl(\tau_*<\infty,\ \tau_*<\tau_{\partial D}\bigr)=0.
\]
This completes the proof.
\end{proof}

\subsection{Complete proof of Proposition~\ref{prop:barrier_invariance_paperA}}
\label{appC:proof_proposition_barrier}

\begin{proof}
Let \((K_n)\) and \((\tau_n)\) be the exhaustion and exit times defined in \eqref{eq:appC_exhaustion}--\eqref{eq:appC_tau_n}. By Theorem~\ref{thm:maximal_no_interior_explosion_paperA}, a unique maximal strong solution exists up to \(\tau_*\), and any finite lifetime obstruction can only occur through boundary hitting/approach.

\medskip
\noindent\textbf{Step 1: Localized It\^o estimate for the barrier Lyapunov function.}
Fix \(n\ge1\) and \(t\ge0\). Since \(x_{s\wedge \tau_n}\in K_n\) for all \(s\), It\^o's formula for \(V_b(x_{t\wedge \tau_n})\) yields
\begin{align}
V_b(x_{t\wedge \tau_n})
&=
V_b(x_0)+\int_0^{t\wedge \tau_n}\mathcal LV_b(x_s)\,ds + M_{t\wedge \tau_n}^{(b)} \notag\\
&\le
V_b(x_0)+C_b\int_0^{t\wedge \tau_n}\bigl(1+V_b(x_s)\bigr)\,ds + M_{t\wedge \tau_n}^{(b)},
\label{eq:appC_barrier_ito}
\end{align}
where \(M_{t\wedge \tau_n}^{(b)}\) is a true martingale.

Taking expectations,
\begin{equation}
\label{eq:appC_barrier_expect_pre_gronwall}
\mathbb E[V_b(x_{t\wedge \tau_n})]
\le
V_b(x_0)+C_b t + C_b\int_0^t \mathbb E[V_b(x_{s\wedge \tau_n})]\,ds.
\end{equation}
Applying Gr\"onwall,
\begin{equation}
\label{eq:appC_barrier_expect_bound}
\mathbb E[V_b(x_{t\wedge \tau_n})]
\le e^{C_b t}\bigl(V_b(x_0)+C_b t\bigr),\qquad \forall n\ge1,\ \forall t\ge0.
\end{equation}
Hence
\begin{equation}
\label{eq:appC_barrier_uniform_bound}
\sup_{n\ge1}\mathbb E[V_b(x_{t\wedge \tau_n})]<\infty\qquad \text{for each fixed }t\ge0.
\end{equation}

\medskip
\noindent\textbf{Step 2: Contradiction argument if boundary is hit with positive probability.}
Assume for contradiction that \(\mathbb P(\tau_{\partial D}<\infty)>0\). Then there exists \(T>0\) such that
\[
\mathbb P(\tau_{\partial D}\le T)>0.
\]
On the event \(\{\tau_{\partial D}\le T\}\), continuity of sample paths implies \(x_t\to x_{\tau_{\partial D}}\in\partial D\) as \(t\uparrow \tau_{\partial D}\) from below (or equivalently \(x_t\to \partial D\) within \(D\)). By the barrier blow-up property \textbf{(B1)},
\[
V_b(x_t)\to\infty\qquad \text{as }t\uparrow \tau_{\partial D}\ \text{on }\{\tau_{\partial D}\le T\}.
\]
Because \(\tau_n\uparrow \tau_*\) and Theorem~\ref{thm:maximal_no_interior_explosion_paperA} implies \(\tau_*\ge \tau_{\partial D}\) a.s., the stopping sequence \(T\wedge \tau_n\) approaches \(T\wedge \tau_{\partial D}\) from below on \(\{\tau_{\partial D}\le T\}\). In particular,
\begin{equation}
\label{eq:appC_barrier_divergence}
V_b\bigl(x_{T\wedge \tau_n}\bigr)\to\infty
\qquad\text{on }\{\tau_{\partial D}\le T\}.
\end{equation}
Applying Fatou's lemma to the nonnegative random variables \(V_b(x_{T\wedge \tau_n})\mathbf 1_{\{\tau_{\partial D}\le T\}}\), we obtain
\begin{align}
\infty
&=
\mathbb E\!\left[\liminf_{n\to\infty}V_b(x_{T\wedge \tau_n})\mathbf 1_{\{\tau_{\partial D}\le T\}}\right] \notag\\
&\le
\liminf_{n\to\infty}\mathbb E\!\left[V_b(x_{T\wedge \tau_n})\mathbf 1_{\{\tau_{\partial D}\le T\}}\right]
\le
\liminf_{n\to\infty}\mathbb E[V_b(x_{T\wedge \tau_n})],
\label{eq:appC_fatou_contradiction}
\end{align}
contradicting the uniform bound \eqref{eq:appC_barrier_uniform_bound} at time \(T\).

Therefore \(\mathbb P(\tau_{\partial D}\le T)=0\) for every \(T>0\), hence
\[
\mathbb P(\tau_{\partial D}<\infty)=0.
\]

\medskip
\noindent\textbf{Step 3: Globality and positivity invariance.}
Since boundary hitting is a.s.\ impossible and Theorem~\ref{thm:maximal_no_interior_explosion_paperA} excludes interior breakdown before boundary hitting, the maximal solution extends for all times:
\[
\tau_*=\infty\qquad\text{a.s.}
\]
Moreover,
\[
x_t\in D\qquad \text{for all }t\ge0\quad\text{a.s.}
\]
This proves positivity invariance.
\end{proof}

\subsection{Remarks on verifiability of the assumptions in the present model}
\label{appC:assumption_verification_remarks}

\begin{remark}[Local regularity of \(b\) and \(a\) on \(D\)]
\label{rmk:appC_ba_smooth}
For the R--M demographic diffusion constructed in Sections~\ref{sec:ctmc_cme_paperA}--\ref{sec:diffusion_covariance_paperA}, the drift \(b(x)\) and covariance \(a(x)\) are polynomial/rational combinations of \(N,P\) with denominator \(1+N\). Since \(N>0\) on \(D=(0,\infty)^2\), the denominator never vanishes on \(D\), and therefore \(b\) and \(a\) are smooth on \(D\) (in particular, locally Lipschitz on every \(K\Subset D\)).
\end{remark}

\begin{remark}[Local regularity of common factorizations \(\Sigma\)]
\label{rmk:appC_sigma_local_lipschitz}
Two coefficient factorizations used in the paper fit assumption \textbf{(A1)} on precompact subsets \(K\Subset D\):
\begin{enumerate}[label=(\roman*)]
\item Event-based factorization.
If
\[
\Sigma_K(x)=\frac{1}{\sqrt\Omega}\,S\,\sqrt{\mathrm{diag}(f(x))},
\]
then on \(K\Subset D\), each intensity \(f_j(x)\) is smooth and bounded away from singularities; in the present model the relevant \(f_j\) are nonnegative smooth functions on \(D\). Hence \(\sqrt{f_j(x)}\) is locally Lipschitz on \(K\) (in particular when \(f_j\) is bounded away from \(0\), and more generally for the explicit polynomial/rational intensities used here on compacts).
\item Matrix square-root factorization.
If \(\Sigma(x)\) is chosen as a measurable square root of \(a(x)\) (e.g., Cholesky where \(a(x)\) is positive definite), then one typically works locally on compact subsets on which rank/signature does not change. On such sets, \(\Sigma\) can be chosen locally Lipschitz. For the well-posedness statements in Section~\ref{sec:wellposed_paperA}, only local Lipschitz regularity on each \(K\Subset D\) is required.
\end{enumerate}
\end{remark}

\begin{remark}[Why boundary degeneracy does not conflict with the open-domain theorem]
\label{rmk:appC_boundary_degeneracy}
Demographic diffusions are often degenerate at \(\partial D\) because reaction intensities vanish as \(N\downarrow0\) and/or \(P\downarrow0\). This does not contradict assumption \textbf{(A1)}, which is imposed only on precompact subsets \(K\Subset D\), i.e.\ strictly away from the boundary. The boundary behavior is handled separately through the maximal-solution formulation and, when needed, the barrier-Lyapunov condition of Proposition~\ref{prop:barrier_invariance_paperA}.
\end{remark}

\subsection{Abstract lemma: two-stage Lyapunov architecture on open domains}
\label{appC:abstract_template}

The arguments above can be summarized in a reusable form for boundary-degenerate ecological diffusions on open domains.

\begin{lemma}[Two-stage Lyapunov template on an open domain]
\label{lem:appC_two_stage_template}
Let \(D\subset\mathbb R^d\) be open, and consider
\[
dX_t=b(X_t)\,dt+\Sigma(X_t)\,dW_t,\qquad X_0\in D,
\]
with coefficients locally Lipschitz and locally bounded on every precompact \(K\Subset D\). Assume:
\begin{enumerate}[label=(\alph*)]
\item there exists \(V_\infty\in C^2(D;[0,\infty))\) with \(V_\infty(x)\to\infty\) as \(\|x\|\to\infty\) in \(D\), and
\[
\mathcal LV_\infty\le C_\infty(1+V_\infty);
\]
\item optionally, there exists \(V_b\in C^2(D;[0,\infty))\) with \(V_b(x)\to\infty\) as \(x\to\partial D\) within \(D\), and (locally on an exhaustion) 
\[
\mathcal LV_b\le C_b(1+V_b).
\]
\end{enumerate}
Then:
\begin{enumerate}[label=(\roman*),leftmargin=1.8em]
\item under (a), there exists a unique maximal strong solution, and no loss of well-posedness occurs in the interior before boundary hitting;
\item under (a)+(b), the boundary is a.s.\ not hit in finite time, hence the solution is global and remains in \(D\) for all \(t\ge0\).
\end{enumerate}
\end{lemma}

\begin{proof}
Part (i) follows the localization/patching and \(V_\infty\)-Lyapunov estimate used in the proof of Theorem~\ref{thm:maximal_no_interior_explosion_paperA}; part (ii) follows by the barrier contradiction argument used in the proof of Proposition~\ref{prop:barrier_invariance_paperA}.
\end{proof}

Lemma~\ref{lem:appC_two_stage_template} is not needed for any statement in the main text, but it clarifies the general structure of the two-stage Lyapunov architecture and may be useful for extensions to other demographic diffusions and food-web models.

\end{appendices}

\bibliography{reference}

@article{grunert_evolutionarily_2021,
	title = {Evolutionarily stable strategies in stable and periodically fluctuating populations: the rosenzweig{\textendash}{MacArthur} predator{\textendash}prey model},
	volume = {118},
	issn = {0027-8424, 1091-6490},
	shorttitle = {Evolutionarily stable strategies in stable and periodically fluctuating populations},
	doi = {10.1073/pnas.2017463118},
	language = {en},
	number = {4},
	journal = {Proceedings of the National Academy of Sciences},
	author = {Grunert, K. and Holden, H. and Jakobsen, E. R. and Stenseth, N. C.},
	year = {2021},
	pages = {e2017463118},
}

@article{beay_stability_2019,
	title = {Stability of a stage-structure rosenzweig-{MacArthur} model incorporating holling type-{II} functional response},
	volume = {546},
	issn = {1757-8981, 1757-899X},
	doi = {10.1088/1757-899X/546/5/052017},
	number = {5},
	journal = {IOP Conference Series: Materials Science and Engineering},
	author = {Beay, L. K. and Suryanto, A. and Darti, I. and {Trisilowati}},
	year = {2019},
	pages = {052017},
}

@article{gard_persistence_1984,
	title = {Persistence in stochastic food web models},
	volume = {46},
	copyright = {http://www.springer.com/tdm},
	issn = {0092-8240, 1522-9602},
	doi = {10.1007/BF02462011},
	language = {en},
	number = {3},
	journal = {Bulletin of Mathematical Biology},
	author = {Gard, T. C.},
	year = {1984},
	pages = {357--370},
}

@article{braumann_ito_2007,
	title = {It{\^o} versus {Stratonovich} calculus in random population growth},
	volume = {206},
	copyright = {https://www.elsevier.com/tdm/userlicense/1.0/},
	issn = {00255564},
	doi = {10.1016/j.mbs.2004.09.002},
	language = {en},
	number = {1},
	journal = {Mathematical Biosciences},
	author = {Braumann, C. A.},
	year = {2007},
	pages = {81--107},
}

@article{yuan_bifurcation_2023,
	title = {Bifurcation and chaotic behavior in stochastic {Rosenzweig}{\textendash}{MacArthur} prey{\textendash}predator model with non-{Gaussian} stable {L{\'e}vy} noise},
	volume = {150},
	issn = {00207462},
	doi = {10.1016/j.ijnonlinmec.2022.104339},
	language = {en},
	journal = {International Journal of Non-Linear Mechanics},
	author = {Yuan, S. and Wang, Z.},
	year = {2023},
	pages = {104339},
}

@article{barraquand_no_2023,
	title = {No sensitivity to functional forms in the {Rosenzweig}{\textendash}{MacArthur} model with strong environmental stochasticity},
	volume = {572},
	issn = {00225193},
	doi = {10.1016/j.jtbi.2023.111566},
	language = {en},
	journal = {Journal of Theoretical Biology},
	author = {Barraquand, F.},
	year = {2023},
	pages = {111566},
}

@article{huang_stochastic_2021,
	title = {A stochastic predator{\textendash}prey model with {Holling} {II} increasing function in the predator},
	volume = {15},
	issn = {1751-3758, 1751-3766},
	doi = {10.1080/17513758.2020.1859146},
	language = {en},
	number = {1},
	journal = {Journal of Biological Dynamics},
	author = {Huang, Y. and Shi, W. and Wei, C. and Zhang, S.},
	year = {2021},
	pages = {1--18},
}

@article{du_conditions_2016,
	title = {Conditions for permanence and ergodicity of certain stochastic predator{\textendash}prey models},
	volume = {53},
	copyright = {https://www.cambridge.org/core/terms},
	issn = {0021-9002, 1475-6072},
	doi = {10.1017/jpr.2015.18},
	language = {en},
	number = {1},
	journal = {Journal of Applied Probability},
	author = {Du, N. H. and Nguyen, D. H. and Yin, G. G.},
	year = {2016},
	pages = {187--202},
}

@article{qi_threshold_2021,
	title = {Threshold behavior of a stochastic predator{\textendash}prey system with prey refuge and fear effect},
	volume = {113},
	issn = {08939659},
	doi = {10.1016/j.aml.2020.106846},
	language = {en},
	journal = {Applied Mathematics Letters},
	author = {Qi, H. and Meng, X.},
	year = {2021},
	pages = {106846},
}

@article{yasin_spatio-temporal_2023,
	title = {Spatio-temporal numerical modeling of stochastic predator-prey model},
	volume = {13},
	issn = {2045-2322},
	doi = {10.1038/s41598-023-28324-6},
	language = {en},
	number = {1},
	journal = {Scientific Reports},
	author = {Yasin, M. W. and Ahmed, N. and Iqbal, M.S. and Raza, A. and Rafiq, M. and Eldin, E. M. T.},
	year = {2023},
	pages = {1990},
}

@article{liu_dynamics_2018,
	title = {Dynamics of a stochastic predator{\textendash}prey model with stage structure for predator and holling type {II} functional response},
	volume = {28},
	issn = {0938-8974, 1432-1467},
	doi = {10.1007/s00332-018-9444-3},
	language = {en},
	number = {3},
	journal = {Journal of Nonlinear Science},
	author = {Liu, Q. and Jiang, D. and Hayat, T. and Alsaedi, A.},
	year = {2018},
	pages = {1151--1187},
}

@article{wu_stochastic_2019,
	title = {Stochastic sensitivity analysis of noise-induced transitions in a predator-prey model with environmental toxins},
	volume = {16},
	issn = {1547-1063},
	doi = {10.3934/mbe.2019104},
	language = {en},
	number = {4},
	journal = {Mathematical Biosciences and Engineering},
	author = {Wu, D. and Wang, H. and Yuan, S.},
	year = {2019},
	pages = {2141--2153},
}

@article{chatterjee_predatorprey_2024,
	title = {A predator{\textendash}prey model with prey refuge: under a stochastic and deterministic environment},
	volume = {112},
	issn = {0924-090X, 1573-269X},
	shorttitle = {A predator{\textendash}prey model with prey refuge},
	doi = {10.1007/s11071-024-09756-9},
	language = {en},
	number = {15},
	journal = {Nonlinear Dynamics},
	author = {Chatterjee, A. and Abbasi, M. A. and Venturino, E. and Zhen, J. and Haque, M.},
	year = {2024},
	pages = {13667--13693},
}

@article{das_modelling_2021,
	title = {Modelling the effect of resource subsidy on a two-species predator-prey system under the influence of environmental noises},
	volume = {9},
	issn = {2195-268X, 2195-2698},
	doi = {10.1007/s40435-020-00750-8},
	language = {en},
	number = {4},
	journal = {International Journal of Dynamics and Control},
	author = {Das, A. and Samanta, G. P.},
	year = {2021},
	pages = {1800--1817},
}

@article{dobramysl_environmental_2013,
	title = {Environmental versus demographic variability in stochastic predator{\textendash}prey models},
	volume = {2013},
	copyright = {http://iopscience.iop.org/info/page/text-and-data-mining},
	issn = {1742-5468},
	doi = {10.1088/1742-5468/2013/10/P10001},
	number = {10},
	journal = {Journal of Statistical Mechanics: Theory and Experiment},
	author = {Dobramysl, U. and T{\"a}uber, U. C.},
	year = {2013},
	pages = {P10001},
}

@inproceedings{naghnaeian_robust_2017,
	address = {Mauna Lani Resort, HI, USA},
	title = {Robust moment closure method for the chemical master equation},
	isbn = {978-1-5090-2182-6},
	doi = {10.1109/CCTA.2017.8062585},
	booktitle = {2017 {IEEE} {Conference} on {Control} {Technology} and {Applications} ({CCTA})},
	publisher = {IEEE},
	author = {Naghnaeian, M. and Del Vecchio, D.},
	year = {2017},
	pages = {967--972},
}

@article{wang_analysis_2025,
	title = {Analysis and mean-field limit of a hybrid {PDE}-{ABM} modeling angiogenesis-regulated resistance evolution},
	volume = {13},
	issn = {2227-7390},
	doi = {10.3390/math13172898},
	language = {en},
	number = {17},
	journal = {Mathematics},
	author = {Wang, L. S. and Yu, J. and Li, S. and Liu, Z.},
	year = {2025},
	pages = {2898},
}

@article{butler_fluctuation-driven_2011,
	title = {Fluctuation-driven turing patterns},
	volume = {84},
	copyright = {http://link.aps.org/licenses/aps-default-license},
	issn = {1539-3755, 1550-2376},
	doi = {10.1103/PhysRevE.84.011112},
	language = {en},
	number = {1},
	journal = {Physical Review E},
	author = {Butler, T. and Goldenfeld, N.},
	year = {2011},
	pages = {011112},
}

@article{bashkirtseva_stochastic_2018,
	title = {Stochastic sensitivity analysis of noise-induced extinction in the ricker model with delay and {Allee} effect},
	volume = {80},
	issn = {0092-8240, 1522-9602},
	doi = {10.1007/s11538-018-0422-6},
	language = {en},
	number = {6},
	journal = {Bulletin of Mathematical Biology},
	author = {Bashkirtseva, I. and Ryashko, L.},
	year = {2018},
	pages = {1596--1614},
}

@article{bashkirtseva_stochastic_2004,
	title = {Stochastic sensitivity of {3D}-cycles},
	volume = {66},
	copyright = {https://www.elsevier.com/tdm/userlicense/1.0/},
	issn = {03784754},
	doi = {10.1016/j.matcom.2004.02.021},
	language = {en},
	number = {1},
	journal = {Mathematics and Computers in Simulation},
	author = {Bashkirtseva, I. A. and Ryashko, L. B.},
	year = {2004},
	pages = {55--67},
}

@article{bashkirtseva_stochastic_2014,
	title = {Stochastic sensitivity of the closed invariant curves for discrete-time systems},
	volume = {410},
	issn = {03784371},
	doi = {10.1016/j.physa.2014.05.037},
	language = {en},
	journal = {Physica A: Statistical Mechanics and its Applications},
	author = {Bashkirtseva, I. and Ryashko, L.},
	year = {2014},
	pages = {236--243},
}

@article{carpenter_early_2011,
	title = {Early warnings of regime shifts: a whole-ecosystem experiment},
	volume = {332},
	copyright = {http://www.sciencemag.org/site/feature/contribinfo/prep/license.xhtml},
	issn = {0036-8075, 1095-9203},
	shorttitle = {Early warnings of regime shifts},
	doi = {10.1126/science.1203672},
	language = {en},
	number = {6033},
	journal = {Science},
	author = {Carpenter, S. R. and Cole, J. J. and Pace, M. L. and Batt, R. and Brock, W. A. and Cline, T.},
	year = {2011},
	pages = {1079--1082},
}

@article{lenton_early_2011,
	title = {Early warning of climate tipping points},
	volume = {1},
	copyright = {http://www.springer.com/tdm},
	issn = {1758-678X, 1758-6798},
	doi = {10.1038/nclimate1143},
	language = {en},
	number = {4},
	journal = {Nature Climate Change},
	author = {Lenton, T. M.},
	year = {2011},
	pages = {201--209},
}

@article{boettiger_no_2013,
	title = {No early warning signals for stochastic transitions: insights from large deviation theory},
	volume = {280},
	issn = {0962-8452, 1471-2954},
	shorttitle = {No early warning signals for stochastic transitions},
	doi = {10.1098/rspb.2013.1372},
	language = {en},
	number = {1766},
	urldate = {2026-02-23},
	journal = {Proceedings of the Royal Society B: Biological Sciences},
	author = {Boettiger, C. and Hastings, A.},
	year = {2013},
	pages = {20131372},
}

@article{feller_diffusion_1954,
	title = {Diffusion processes in one dimension},
	volume = {77},
	issn = {0002-9947, 1088-6850},
	doi = {10.1090/S0002-9947-1954-0063607-6},
	language = {en},
	number = {1},
	journal = {Transactions of the American Mathematical Society},
	author = {Feller, W.},
	year = {1954},
	pages = {1--31},
}

@incollection{peskir_boundary_2015,
	address = {Cham},
	title = {On boundary behaviour of one-dimensional diffusions: from brown to feller and beyond},
	isbn = {978-3-319-16855-5 978-3-319-16856-2},
	shorttitle = {On boundary behaviour of one-dimensional diffusions},
	url = {https://personalpages.manchester.ac.uk/staff/goran.peskir/diffusions.pdf},
	language = {en},
	urldate = {2026-02-23},
	booktitle = {Selected {Papers} {II}},
	publisher = {Springer International Publishing},
	author = {Peskir, G.},
	year = {2015},
	pages = {77--93},
}

@article{pineda-krch_gillespiessa_2009,
	title = {{GillespieSSA}: {A} user-friendly stochastic simulation package for {R}},
	issn = {1756-0357},
	shorttitle = {{GillespieSSA}},
	doi = {10.1038/npre.2009.3673.1},
	language = {en},
	journal = {Nature Precedings},
	author = {Pineda-Krch, M.},
	year = {2009},
}

@article{meyn_stability_1993,
	title = {Stability of markovian processes {III}: foster{\textendash}lyapunov criteria for continuous-time processes},
	volume = {25},
	copyright = {https://www.cambridge.org/core/terms},
	issn = {0001-8678, 1475-6064},
	shorttitle = {Stability of markovian processes {III}},
	doi = {10.2307/1427522},
	language = {en},
	number = {3},
	journal = {Advances in Applied Probability},
	author = {Meyn, S. P. and Tweedie, R. L.},
	year = {1993},
	pages = {518--548},
}

@article{belabbas_rich_2021,
	title = {Rich dynamics in a stochastic predator-prey model with protection zone for the prey and multiplicative noise applied on both species},
	volume = {106},
	issn = {0924-090X, 1573-269X},
	doi = {10.1007/s11071-021-06903-4},
	language = {en},
	number = {3},
	journal = {Nonlinear Dynamics},
	author = {Belabbas, M. and Ouahab, A. and Souna, F.},
	year = {2021},
	pages = {2761--2780},
}

@article{cui_long-term_2026,
	title = {The long-term dynamic behaviors of a stochastic predator-prey model with fear and group defense},
	volume = {34},
	issn = {1531-3492, 1553-524X},
	doi = {10.3934/dcdsb.2025166},
	number = {0},
	journal = {Discrete and Continuous Dynamical Systems - B},
	author = {Cui, Y. and Huang, S. and Li, X. and Shi, S.},
	year = {2026},
	pages = {106--140},
}

@article{zhao_stochastic_2024,
	title = {Stochastic dynamics of coral reef system with stage-structure for crown-of-thorns starfish},
	volume = {181},
	issn = {09600779},
	doi = {10.1016/j.chaos.2024.114629},
	language = {en},
	journal = {Chaos, Solitons \& Fractals},
	author = {Zhao, X. and Liu, L. and Liu, M. and Fan, M.},
	year = {2024},
	pages = {114629},
}

@article{zhan_dynamical_2024,
	title = {Dynamical behavior of a stochastic non-autonomous distributed delay heroin epidemic model with regime-switching},
	volume = {184},
	issn = {09600779},
	doi = {10.1016/j.chaos.2024.115024},
	language = {en},
	journal = {Chaos, Solitons \& Fractals},
	author = {Zhan, J. and Wei, Y.},
	year = {2024},
	pages = {115024},
}

@article{liu_asymptotic_2020,
	title = {Asymptotic stability of a stochastic {May} mutualism system},
	volume = {79},
	issn = {08981221},
	doi = {10.1016/j.camwa.2019.07.022},
	language = {en},
	number = {3},
	journal = {Computers \& Mathematics with Applications},
	author = {Liu, G. and Qi, H. and Chang, Z. and Meng, X.},
	year = {2020},
	pages = {735--745},
}

@article{arditi_coupling_1989,
	title = {Coupling in predator-prey dynamics: {Ratio}-dependence},
	volume = {139},
	copyright = {https://www.elsevier.com/tdm/userlicense/1.0/},
	issn = {00225193},
	shorttitle = {Coupling in predator-prey dynamics},
	doi = {10.1016/S0022-5193(89)80211-5},
	language = {en},
	number = {3},
	journal = {Journal of Theoretical Biology},
	author = {Arditi, R. and Ginzburg, L. R.},
	year = {1989},
	pages = {311--326},
}

@book{allen_introduction_2010,
	address = {Boca Raton, FL},
	edition = {2},
	title = {An introduction to stochastic processes with applications to biology},
	isbn = {978-0-429-18460-4},
	doi = {10.1201/b12537},
	language = {en},
	publisher = {Chapman and Hall/CRC},
	author = {Allen, L. J. S.},
	year = {2010},
}

@article{alonso_stochastic_2007,
	title = {Stochastic amplification in epidemics},
	volume = {4},
	copyright = {https://royalsociety.org/journals/ethics-policies/data-sharing-mining/},
	issn = {1742-5689, 1742-5662},
	doi = {10.1098/rsif.2006.0192},
	language = {en},
	number = {14},
	journal = {Journal of The Royal Society Interface},
	author = {Alonso, D. and McKane, A. J. and Pascual, M.},
	year = {2007},
	pages = {575--582},
}

@article{butler_robust_2009,
	title = {Robust ecological pattern formation induced by demographic noise},
	volume = {80},
	copyright = {http://link.aps.org/licenses/aps-default-license},
	issn = {1539-3755, 1550-2376},
	doi = {10.1103/PhysRevE.80.030902},
	language = {en},
	number = {3},
	journal = {Physical Review E},
	author = {Butler, T. and Goldenfeld, N.},
	year = {2009},
	pages = {030902},
}

@article{boettiger_quantifying_2012,
	title = {Quantifying limits to detection of early warning for critical transitions},
	volume = {9},
	issn = {1742-5689, 1742-5662},
	doi = {10.1098/rsif.2012.0125},
	language = {en},
	number = {75},
	journal = {Journal of The Royal Society Interface},
	author = {Boettiger, C. and Hastings, A.},
	year = {2012},
	pages = {2527--2539},
}

@article{black_stochastic_2012,
	title = {Stochastic formulation of ecological models and their applications},
	volume = {27},
	copyright = {https://www.elsevier.com/tdm/userlicense/1.0/},
	issn = {01695347},
	doi = {10.1016/j.tree.2012.01.014},
	language = {en},
	number = {6},
	journal = {Trends in Ecology \& Evolution},
	author = {Black, A. J. and McKane, A. J.},
	year = {2012},
	pages = {337--345},
}

@article{braumann_variable_2002,
	title = {Variable effort harvesting models in random environments: generalization to density-dependent noise intensities},
	volume = {177-178},
	copyright = {https://www.elsevier.com/tdm/userlicense/1.0/},
	issn = {00255564},
	shorttitle = {Variable effort harvesting models in random environments},
	doi = {10.1016/S0025-5564(01)00110-9},
	language = {en},
	journal = {Mathematical Biosciences},
	author = {Braumann, C. A.},
	year = {2002},
	pages = {229--245},
}

@article{biancalani_noise-induced_2012,
	title = {Noise-induced metastability in biochemical networks},
	volume = {86},
	copyright = {http://link.aps.org/licenses/aps-default-license},
	issn = {1539-3755, 1550-2376},
	doi = {10.1103/PhysRevE.86.010106},
	language = {en},
	number = {1},
	journal = {Physical Review E},
	author = {Biancalani, T. and Rogers, T. and McKane, A. J.},
	year = {2012},
	pages = {010106},
}

@article{cattiaux_quasi-stationary_2009,
	title = {Quasi-stationary distributions and diffusion models in population dynamics},
	volume = {37},
	issn = {0091-1798},
	doi = {10.1214/09-AOP451},
	number = {5},
	journal = {The Annals of Probability},
	author = {Cattiaux, P. and Collet, P. and Lambert, A. and Mart{\'i}nez, S. and M{\'e}l{\'e}ard, S. and San Mart{\'i}n, J.},
	year = {2009},
}

@article{cox_theory_1985,
	title = {A theory of the term structure of interest rates},
	volume = {53},
	issn = {00129682},
	doi = {10.2307/1911242},
	number = {2},
	journal = {Econometrica},
	author = {Cox, J. C. and Ingersoll, J. E. and Ross, S. A.},
	year = {1985},
	pages = {385},
}

@article{dakos_slowing_2008,
	title = {Slowing down as an early warning signal for abrupt climate change},
	volume = {105},
	issn = {0027-8424, 1091-6490},
	doi = {10.1073/pnas.0802430105},
	language = {en},
	number = {38},
	journal = {Proceedings of the National Academy of Sciences},
	author = {Dakos, V. and Scheffer, M. and Van Nes, E. H. and Brovkin, V. and Petoukhov, V. and Held, H.},
	year = {2008},
	pages = {14308--14312},
}

@article{elf_fast_2003,
	title = {Fast evaluation of fluctuations in biochemical networks with the linear noise approximation},
	volume = {13},
	issn = {1088-9051},
	doi = {10.1101/gr.1196503},
	language = {en},
	number = {11},
	journal = {Genome Research},
	author = {Elf, J. and Ehrenberg, M.},
	month = nov,
	year = {2003},
	pages = {2475--2484},
}

@book{ethier_markov_1986,
	address = {New York},
	edition = {1},
	series = {Wiley {Series} in {Probability} and {Statistics}},
	title = {Markov {Processes}: {Characterization} and {Convergence}},
	copyright = {http://doi.wiley.com/10.1002/tdm\_license\_1.1},
	isbn = {978-0-471-08186-9 978-0-470-31665-8},
	shorttitle = {Markov {Processes}},
	doi = {10.1002/9780470316658},
	language = {en},
	publisher = {Wiley},
	author = {Ethier, S. N. and Kurtz, T. G.},
	month = mar,
	year = {1986},
}

@book{ewens_mathematical_2004,
	address = {New York, NY},
	series = {Interdisciplinary {Applied} {Mathematics}},
	title = {Mathematical population genetics},
	volume = {27},
	copyright = {http://www.springer.com/tdm},
	isbn = {978-1-4419-1898-7 978-0-387-21822-9},
	doi = {10.1007/978-0-387-21822-9},
	publisher = {Springer New York},
	author = {Ewens, W. J.},
	year = {2004},
}

@article{feller_parabolic_1952,
	title = {The parabolic differential equations and the associated semi-groups of transformations},
	volume = {55},
	issn = {0003486X},
	doi = {10.2307/1969644},
	number = {3},
	journal = {The Annals of Mathematics},
	author = {Feller, W.},
	year = {1952},
	pages = {468},
}

@book{gardiner_stochastic_2009,
	address = {Berlin Heidelberg},
	edition = {4th ed},
	series = {Springer series in synergetics},
	title = {Stochastic methods: a handbook for the natural and social sciences},
	isbn = {978-3-642-08962-6 978-3-540-70712-7},
	shorttitle = {Stochastic methods},
	language = {eng},
	number = {13},
	publisher = {Springer},
	author = {Gardiner, C. W.},
	year = {2009},
}

@article{gilpin_enriched_1972,
	title = {Enriched predator-prey systems: theoretical stability},
	volume = {177},
	issn = {0036-8075, 1095-9203},
	shorttitle = {Enriched predator-prey systems},
	doi = {10.1126/science.177.4052.902},
	language = {en},
	number = {4052},
	journal = {Science},
	author = {Gilpin, M. E.},
	year = {1972},
	pages = {902--904},
}

@article{gillespie_chemical_2000,
	title = {The chemical {Langevin} equation},
	volume = {113},
	issn = {0021-9606, 1089-7690},
	doi = {10.1063/1.481811},
	language = {en},
	number = {1},
	journal = {The Journal of Chemical Physics},
	author = {Gillespie, D. T.},
	year = {2000},
	pages = {297--306},
}

@inproceedings{hofrichter_diffusion_2014,
	title = {On the diffusion approximation of wright{\textendash}fisher models with several alleles and loci and its geometry},
	url = {https://api.semanticscholar.org/CorpusID:2240497},
	urldate = {2026-02-25},
	author = {Hofrichter, J.},
	year = {2014},
}

@book{khasminskii_stochastic_2012,
	address = {Berlin, Heidelberg},
	series = {Stochastic {Modelling} and {Applied} {Probability}},
	title = {Stochastic stability of differential equations},
	volume = {66},
	copyright = {https://www.springernature.com/gp/researchers/text-and-data-mining},
	isbn = {978-3-642-23279-4 978-3-642-23280-0},
	doi = {10.1007/978-3-642-23280-0},
	language = {en},
	publisher = {Springer Berlin Heidelberg},
	author = {Khasminskii, R.},
	year = {2012},
}

@article{liu_survival_2011,
	title = {Survival analysis of stochastic competitive models in a polluted environment and stochastic competitive exclusion principle},
	volume = {73},
	copyright = {http://www.springer.com/tdm},
	issn = {0092-8240, 1522-9602},
	doi = {10.1007/s11538-010-9569-5},
	language = {en},
	number = {9},
	journal = {Bulletin of Mathematical Biology},
	author = {Liu, M. and Wang, K. and Wu, Q.},
	year = {2011},
	pages = {1969--2012},
}

@book{mao_stochastic_2008,
	address = {Chichester},
	edition = {2nd ed},
	title = {Stochastic differential equations and applications},
	isbn = {978-1-904275-34-3},
	publisher = {Horwood Pub},
	author = {Mao, X.},
	year = {2008},
}

@book{may_stability_2001,
	address = {Princeton},
	edition = {1st Princeton landmarks in biology ed},
	series = {Princeton landmarks in biology},
	title = {Stability and complexity in model ecosystems},
	isbn = {978-0-691-08861-7},
	publisher = {Princeton University Press},
	author = {May, R. M.},
	year = {2001},
}

@article{mckane_predator-prey_2005,
	title = {Predator-prey cycles from resonant amplification of demographic stochasticity},
	volume = {94},
	copyright = {http://link.aps.org/licenses/aps-default-license},
	issn = {0031-9007, 1079-7114},
	doi = {10.1103/PhysRevLett.94.218102},
	language = {en},
	number = {21},
	journal = {Physical Review Letters},
	author = {McKane, A. J. and Newman, T. J.},
	month = jun,
	year = {2005},
	pages = {218102},
}

@article{newman_extinction_2004,
	title = {Extinction times and moment closure in the stochastic logistic process},
	volume = {65},
	copyright = {https://www.elsevier.com/tdm/userlicense/1.0/},
	issn = {00405809},
	doi = {10.1016/j.tpb.2003.10.003},
	language = {en},
	number = {2},
	journal = {Theoretical Population Biology},
	author = {Newman, T. J. and Ferdy, J.-B. and Quince, C.},
	year = {2004},
	pages = {115--126},
}

@article{nasell_extinction_2001,
	title = {Extinction and quasi-stationarity in the verhulst logistic model},
	volume = {211},
	copyright = {https://www.elsevier.com/tdm/userlicense/1.0/},
	issn = {00225193},
	doi = {10.1006/jtbi.2001.2328},
	language = {en},
	number = {1},
	journal = {Journal of Theoretical Biology},
	author = {N{\r a}sell, I.},
	year = {2001},
	pages = {11--27},
}

@article{ovaskainen_stochastic_2010,
	title = {Stochastic models of population extinction},
	volume = {25},
	copyright = {https://www.elsevier.com/tdm/userlicense/1.0/},
	issn = {01695347},
	doi = {10.1016/j.tree.2010.07.009},
	language = {en},
	number = {11},
	journal = {Trends in Ecology \& Evolution},
	author = {Ovaskainen, O. and Meerson, B.},
	year = {2010},
	pages = {643--652},
}

@article{pinedakrch_tale_2007,
	title = {A tale of two cycles {\textendash} distinguishing quasi-cycles and limit cycles in finite predator{\textendash}prey populations},
	volume = {116},
	issn = {0030-1299, 1600-0706},
	doi = {10.1111/j.2006.0030-1299.14940.x},
	language = {en},
	number = {1},
	journal = {Oikos},
	author = {Pineda-Krch, M. and J. Blok, H. and Dieckmann, U. and Doebeli, M.},
	year = {2007},
	pages = {53--64},
}

@article{roy_stability_2007,
	title = {The stability of ecosystems: {A} brief overview of the paradox of enrichment},
	volume = {32},
	copyright = {http://www.springer.com/tdm},
	issn = {0250-5991, 0973-7138},
	shorttitle = {The stability of ecosystems},
	doi = {10.1007/s12038-007-0040-1},
	language = {en},
	number = {2},
	journal = {Journal of Biosciences},
	author = {Roy, S. and Chattopadhyay, J.},
	year = {2007},
	pages = {421--428},
}

@book{renshaw_modelling_1995,
	address = {Cambridge},
	edition = {1. paperback ed., reprint},
	series = {Cambridge studies in mathematical biology},
	title = {Modelling biological populations in space and time},
	isbn = {978-0-521-44855-0 978-0-521-30388-0},
	number = {11},
	publisher = {Cambridge Univ. Press},
	author = {Renshaw, E.},
	year = {1995},
}

@article{rosenzweig_graphical_1963,
	title = {Graphical representation and stability conditions of predator-prey interactions},
	volume = {97},
	issn = {00030147, 15375323},
	url = {http://www.jstor.org.utk.idm.oclc.org/stable/2458702},
	number = {895},
	urldate = {2026-02-25},
	journal = {The American Naturalist},
	publisher = {[The University of Chicago Press, The American Society of Naturalists]},
	author = {Rosenzweig, M. L. and MacArthur, R. H.},
	year = {1963},
	pages = {209--223},
}

@article{rogers_demographic_2012,
	title = {Demographic noise can lead to the spontaneous formation of species},
	volume = {97},
	issn = {0295-5075, 1286-4854},
	doi = {10.1209/0295-5075/97/40008},
	number = {4},
	journal = {EPL (Europhysics Letters)},
	author = {Rogers, T. and McKane, A. J. and Rossberg, A. G.},
	year = {2012},
	pages = {40008},
}

@article{rosenzweig_paradox_1971,
	title = {Paradox of enrichment: destabilization of exploitation ecosystems in ecological time},
	volume = {171},
	issn = {0036-8075, 1095-9203},
	shorttitle = {Paradox of enrichment},
	doi = {10.1126/science.171.3969.385},
	language = {en},
	number = {3969},
	journal = {Science},
	author = {Rosenzweig, M. L.},
	year = {1971},
	pages = {385--387},
}

@article{scheffer_early-warning_2009,
	title = {Early-warning signals for critical transitions},
	volume = {461},
	copyright = {https://www.springer.com/tdm},
	issn = {0028-0836, 1476-4687},
	doi = {10.1038/nature08227},
	language = {en},
	number = {7260},
	journal = {Nature},
	author = {Scheffer, M. and Bascompte, J. and Brock, W. A. and Brovkin, V. and Carpenter, S. R. and Dakos, V.},
	year = {2009},
	pages = {53--59},
}

@book{van_kampen_stochastic_2007,
	address = {Amsterdam ; Boston},
	edition = {3rd ed},
	series = {North-{Holland} personal library},
	title = {Stochastic processes in physics and chemistry},
	isbn = {978-0-444-52965-7},
	publisher = {Elsevier},
	author = {van Kampen, N. G.},
	year = {2007},
}

@article{wang_analysis_2025-1,
	title = {Analysis framework for stochastic predator{\textendash}prey model with demographic noise},
	volume = {27},
	issn = {25900374},
	doi = {10.1016/j.rinam.2025.100621},
	language = {en},
	journal = {Results in Applied Mathematics},
	author = {Wang, L. S. and Yu, J.},
	year = {2025},
	pages = {100621},
}

@article{zhao_survival_2015,
	title = {Survival and stationary distribution analysis of a stochastic competitive model of three species in a polluted environment},
	volume = {77},
	issn = {0092-8240, 1522-9602},
	doi = {10.1007/s11538-015-0086-4},
	language = {en},
	number = {7},
	journal = {Bulletin of Mathematical Biology},
	author = {Zhao, Y. and Yuan, S. and Ma, J.},
	year = {2015},
	pages = {1285--1326},
}

@article{liu_bidirectional_2025,
	title = {Bidirectional endothelial feedback drives turing-vascular patterning and drug-resistance niches: a hybrid {PDE}-agent-based study},
	volume = {12},
	issn = {2306-5354},
	shorttitle = {Bidirectional endothelial feedback drives turing-vascular patterning and drug-resistance niches},
	doi = {10.3390/bioengineering12101097},
	language = {en},
	number = {10},
	journal = {Bioengineering},
	author = {Liu, Z. and Wang, L. S. and Yu, J. and Zhang, J. and Martel, E. and Li, S.},
	year = {2025},
	pages = {1097},
}

@article{liang_global_2025,
	title = {Global well-posedness and stability of nonlocal damage-structured lineage model with feedback and dedifferentiation},
	volume = {13},
	issn = {2227-7390},
	doi = {10.3390/math13223583},
	language = {en},
	number = {22},
	journal = {Mathematics},
	author = {Liang, Y. and Wang, L. S. and Yu, J. and Liu, Z.},
	year = {2025},
	pages = {3583},
}

@article{wang_algebraicspectral_2026,
	title = {Algebraic{\textendash}spectral thresholds and discrete{\textendash}continuous stability transfer in {Leslie}{\textendash}{Gower} systems},
	volume = {34},
	issn = {2688-1594},
	doi = {10.3934/era.2026013},
	number = {1},
	journal = {Electronic Research Archive},
	author = {Wang, L. S. and Yu, J.},
	year = {2026},
	pages = {251--290},
}

@article{wang_damage-structured_2026,
	title = {A damage-structured {PDE} model of stem cell hierarchies: {The} dual role of dedifferentiation in tissue homeostasis and aging},
	volume = {21},
	issn = {1932-6203},
	shorttitle = {A damage-structured {PDE} model of stem cell hierarchies},
	doi = {10.1371/journal.pone.0335163},
	language = {en},
	number = {2},
	journal = {PLOS One},
	author = {Wang, L. S. and Yu, J. and Liu, Z.},
	editor = {Ghazimoradi, Mohammad H.},
	year = {2026},
	pages = {e0335163},
}

\end{document}